\documentclass[a4paper]{article}
\usepackage{a4wide}
\usepackage{amsmath,amsfonts,amssymb,amsthm}
\usepackage{color}
\newcommand{\R}{{\mathbb R}}

\newcommand{\Tr}{{\rm Tr}}

\newcommand{\x}{{\bf x}}

\newtheorem{theorem}{Theorem}[section]

\newtheorem{proposition}[theorem]{Proposition}
\newtheorem{corollary}[theorem]{Corollary}
\newtheorem{lemma}[theorem]{Lemma}

\theoremstyle{definition}
\newtheorem{remark}[theorem]{Remark}

\numberwithin{equation}{section}

\begin{document}

\noindent 
\begin{center}
\textbf{\large Adiabatically switched-on electrical bias in continuous systems, and the Landauer-B{\"u}ttiker formula}
\end{center}

\begin{center}
October 4th, 2007
\end{center}

\vspace{0.5cm}

\noindent 

\begin{center}
\small{ 
Horia D. Cornean\footnote{Department of Mathematical Sciences, 
    Aalborg
    University, Fredrik Bajers Vej 7G, 9220 Aalborg, Denmark; e-mail:
    cornean@math.aau.dk},
Pierre Duclos\footnote{Centre de Physique Th\'eorique, Campus de Luminy, Case 907
13288 Marseille cedex 9, France }${}^,$
\footnote{Universit{\'e} du sud Toulon-Var, Toulon, France; 
e-mail: duclos@univ-tln.fr}, 
Gheorghe Nenciu\footnote{Dept. Theor. Phys.,
University of Bucharest, P. O. Box MG11, 
RO-76900 Bu\-cha\-rest, Romania; }${}^,$
\footnote{
Inst. of Math. ``Simion Stoilow'' of
the Romanian Academy, P. O. Box 1-764, RO-014700 Bu\-cha\-rest, Romania;
e-mail: Gheorghe.Nenciu@imar.ro},
Radu Purice
\footnote{
Inst. of Math. ``Simion Stoilow'' of
the Romanian Academy, P. O. Box 1-764, RO-014700 Bu\-cha\-rest, Romania;
e-mail: Radu.Purice@imar.ro}}
     
\end{center}

\vspace{0.5cm}

\noindent

\begin{abstract}
Consider a three dimensional system which looks like a cross-connected
pipe system, i.e. a small sample coupled to a finite number of
leads. We investigate the current running through this system, in
the linear response regime, when we adiabatically turn on an electrical bias
between leads. The main technical tool is the use of a finite volume
regularization, which allows us to define the current coming out of a
lead as the time derivative of its charge. We finally prove that in
virtually all physically interesting situations, the conductivity
tensor is given by a Landauer-B{\"u}ttiker type formula.
\end{abstract}

\vspace{0.5cm}

\section{Introduction}

In any experiment which involves running a current through a microscopic sample,
having a good quantum description of various
conductivity coefficients is a priority. Such a description has been first derived
by Landauer \cite{Landauer} and then generalized by B\"uttiker
\cite{But}; this is what one now calls the Landauer-B\"uttiker
formalism. The main idea is that the conductivity of mesoscopic
samples connected to ideal leads where the carriers are quasi free
fermions, is completely characterized by a one particle scattering
matrix. 

Many people have since contributed to the justification of this
formalism, starting from the first principles of non-equilibrium
quantum statistical mechanics. In this respect, there are two
different ways to model such a conduction problem, and let us briefly
discuss both of them. 

The first approach is the one in which { one starts} with several decoupled semi-infinite
leads, each of them being at equilibrium \cite{JP1}. Let us assume for simplicity
that they are in grand canonical Gibbs states having the same temperature but
different chemical potentials. Then at $t=0$ they are
suddenly joint together with a sample, and the new composed system is allowed to evolve
freely until it reaches a steady state at $t=\infty$. Then one can define a
current as the (Cesaro limit of the) time derivative of a regularized charge operator, and
after lifting this regularization one obtains the L-B formula. This
procedure is by now very well understood and completely solved in a
series of very recent works (see \cite{AJPP}, \cite{Ne1} and references
therein). Let us mention that in this approach a crucial ingredient is the fact
that the perturbation introduced by the sample and coupling 
is localized in space. One can even allow the carriers to interact in
the sample \cite{JOP}, and the theory still works (even though the L-B formula
must be replaced by something more complicated). Note that even if we
choose to adiabatically turn on the coupling
between the semi-infinite leads, the result will be the same. 
A rigorous proof of this fact is in preparation
\cite{CNZ}.  

The second approach is the one in which the leads ({\it long, but
  finite}) are already coupled with the sample, and at $t=-\infty$ the
full system is in a Gibbs 
equilibrium state at a given temperature and chemical potential. Then
we adiabatically turn on a potential bias between the leads, modeling
in this way a gradual appearance of a difference in the chemical
potentials. The statistical density matrix is found as the solution of
a quantum Liouville equation. The current coming out of a given lead is defined to
be the time derivative at $t=0$ of its mean charge. Then one performs 
the linear response approximation with respect to the bias thus obtaining a Kubo-like
formula \cite{BGKS}, {{} and finally} the thermodynamic and adiabatic limits. The
current is given by the same L-B formula, specialized to the linear
response case. Note that in contrast to the previously discussed 
approach, the perturbation introduced by the electrical bias
is not spatially localized, and this makes the
adiabatic limit for the full state (i.e. without the linear response
approximation) rather difficult. 

Some significant papers from the physics literature which initiated the
second approach are \cite{Cini}, \cite{Ste}, \cite{SA}, \cite{FL}, \cite{LA} and \cite{BS}. 
They contain
many nice and interesting physical ideas, even 
for systems which allow local self-interactions, 
but with no real mathematical substance, also due to the fact that
many techniques were not yet available at {{} this} time. 

A first
mathematically sound derivation of the L-B formula on a discrete model
was obtained in \cite{CJM1} and further investigated in
\cite{CJM2}. In the current paper we greatly improve the method of proof of
\cite{CJM1}, which also allows us to extend the results to the continuous
case. 

There are many interesting and hard questions which remain to be
answered. Namely, can one compute the adiabatic limit in the second
approach without the linear response approximation? If yes, then 
what is the connection with the first approach?
Are the two steady states identical? If not, on which class of
observables do they have equal expectations? Can one say anything rigorous about transient currents 
(see \cite{MGM} and references therein)?

A distinct issue is the time dependent coupling and/or bias
introduction, eventually periodic in time, closely related to the
pumping problem. We here mainly refer to the so called BPT formula
\cite{BPT}, rigorously investigated and justified in a series of
papers by Avron {\it et al} \cite{Avron1}, \cite{Avron2}. We also note
that Landauer type formulas are universal in a way, and can be met in
various other contexts like for example in the viscosity experienced by a
topological swimmer \cite{AGO}.

The paper is organized as follows. 

Section 2 introduces the mathematical model, and derives a current
formula \eqref{curent-2} via linear response theory, finite 
volume regularization and adiabatic switching of the electrical bias between leads.

Section 3 contains the proof of the thermodynamic limit (i.e. the
{{} length} of the leads is taken to infinity). The new current formula is
given in \eqref{curent-22}. 

Section 4 contains the adiabatic limit of the current, and makes the connection with
the stationary scattering objects in \eqref{prima55}. 

In Section 5 we derive the usual Landauer-B\"uttiker formula from
\eqref{prima55}; we state our main results in Theorem
\ref{teoremaprincipala}. We also derive a continuity equation in
subsection 5.1, which shows that the steady current is the same, no matter
where we measure it on a given lead.

\section{The finite volume regularization}

Since our results can be easily extended to one and two dimensions, we
will only consider the three dimensional case. The small device through
which the current flows is modelled by a bounded domain of $\R^3$, not
{{} necessarily} simply connected, but with a regular enough boundary. This is
the "sample", which is linked to a finite number of "leads" (which can
be cylinders of {{} length} $L$ and radius $1$). Let us
for simplicity only consider two leads, both cylinders being parallel
to the first coordinate axis and located in $-L-a<x_1<-a$ (the left
lead) and $a<x_1<L+a$ (the right lead). Here $a$ is a nonnegative
number. The sample is a subset of the slab $-a<x_1<a$. Note that one
possible system is just the union of the two
leads, when $a=0$ and the sample is absent. If $a>0$ we demand
the sample to be smoothly pasted to the leads, so that the boundary
of their union is at least $C^2$.

\subsection{The equilibrium state}

The union between the leads and sample will be denoted by
$\mathcal{X}_L$. The one particle Hamilton operator is
given by $\mathbf{H}_L:=-\Delta_L +w$, where $-\Delta_L$ is the
Laplace operator with Dirichlet boundary conditions in
$\mathcal{X}_L$. The potential $w$ is smooth and compactly supported
in the region where the sample is located, i.e. $\mathcal{X}_L\cap
\{\x\in \R^3:\; -a<x_1<a\}$. Without loss of generality, we put $w\geq 0$. 
The one particle Hilbert space is $\mathcal{H}^L:=L^2(\mathcal{X}_L)$.

Denote by $\mathbb{H}_D(\mathcal{X}_L):=H^1_0(\mathcal{X}_L)\cap
H^2(\mathcal{X}_L)$, where 
$H^1_0(\mathcal{X}_L)$ and $H^2(\mathcal{X}_L)$ are the usual Sobolev
spaces on the open domain 
$\mathcal{X}_L\subset\R^3$. 
Since we assumed enough regularity on the boundary of $\mathcal{X}_L$,
the operators:
\begin{align}
& \label{Dir-Lapl}
-\Delta_L :\mathbb{H}_D(\mathcal{X}_L)\rightarrow\mathcal{H}^L \\
&\label{Hamilt}
\mathbf{H}_L:=-\Delta_L\,+\,w(Q):\,\mathbb{H}_D(\mathcal{X}_L)\rightarrow\mathcal{H}^L
\end{align} 
are self-adjoint on $\mathbb{H}_D(\mathcal{X}_L)$. It is also very well
known that their spectrum is purely discrete and accumulates to
$+\infty$. Due to our assumptions on $w$, the Hamiltonian $\mathbf{H}_L$ in
(\ref{Hamilt}) is a positive operator and its resolvent is denoted by
\begin{equation}\label{rez}
 \mathbf{R}_L(z)\,:=\,(\mathbf{H}_L\,-\,z)^{-1}
\end{equation} 
for all $z\in\mathbb{C}\setminus\mathbb{R}_+$; we shall simply denote 
$\mathbf{R}_L:=\mathbf{R}_L(-1)$. 

We only consider the grand-canonical ensemble. In the remote past, 
$t\rightarrow-\infty$, the
electron gas is in thermodynamic equilibrium at a 
temperature $T=\frac{1}{\beta}>0$ and a chemical potential $\mu$. The appropriate
framework is the second quantization. Here the Hilbert space is 
$\mathfrak{F}_L:=\underset{n\in\mathbb{N}}{\oplus}(\mathcal{H}^L)^{\wedge
  n}$, where $(\mathcal{H}^L)^{\wedge n}$ is the $n$ times
{{} antisymmetric} tensor product of 
$\mathcal{H}^L$ with itself. Let
$\boldsymbol{\mathfrak{n}}_L=d\Gamma(\boldsymbol{1}_L)$ be the
number operator (where $\boldsymbol{1}_L$ is the identity operator on $\mathcal{H}^L$).

The second quantization of $\boldsymbol{H}_L$ is denoted by 
$\boldsymbol{\mathfrak{h}}_L:=d\Gamma(\mathbf{H}_L)$. Let 
$$
\boldsymbol{\mathfrak{u}}_L(t)\,:=\,\Gamma(e^{-it\mathbf{H}_L})\,=\,e^{-it\boldsymbol{\mathfrak{h}}_L}
$$
be the associated {{} time-}evolution  of the system. Under our conditions, the
operator $e^{-\beta \mathbf{H}_L}\in\mathbb{B}_1[\mathcal{H}^L]$, the ideal
of trace-class operators on $\mathcal{H}^L$, and this property extends
itself to the second quantized operators.

The associated Gibbs equilibrium state is:
\begin{equation}\label{G-state}
\mathfrak{p}_L:=\frac{e^{-\beta(\boldsymbol{\mathfrak{h}}_L-\mu\boldsymbol{\mathfrak{n}}_L)}}
{\Tr\left \{e^{-\beta(\boldsymbol{\mathfrak{h}}_L-\mu\boldsymbol{\mathfrak{n}}_L)}\right\}}.
\end{equation}
It is a standard fact (see Proposition 5.2.23 in \cite{brarob}) that if we denote by 
$\rho(\lambda):=\left(e^{\beta(\lambda-\mu)}+1\right)^{-1}$, then for
any bounded operator 
$\mathbf{T}\in\mathbb{B}[\mathcal{H}^L]$ we have:
\begin{equation}\label{prima2}
\Tr_{\mathfrak{F}_L}\left\{\mathfrak{p}_L\,d\Gamma(\mathbf{T})\right\}\,=\,
\Tr_{\mathcal{H}^L}\left\{\rho(\mathbf{H}_L)\,
\mathbf{T}\right\}.
\end{equation}

\subsection{Turning on the bias}

Let us introduce three projections defined by natural restrictions on
the left lead, sample, and right lead respectively:
\begin{align}\label{adoua2}
& \Pi_-:\mathcal{H}^L\rightarrow L^2(\mathcal{X}_L\cap \{\x\in\R^3:\;
-L-a<x_1<-a\}), 
\nonumber \\
& \Pi_+:\mathcal{H}^L\rightarrow L^2(\mathcal{X}_L\cap \{\x\in\R^3:\; a<x_1<L+a\}),
\nonumber \\
&\Pi_0:\mathcal{H}^L\rightarrow L^2(\mathcal{X}_L\cap \{\x\in\R^3:\; -a<x_1<a\}).
\end{align}
If $v_\pm\in\R$, define the bias between leads as 
\begin{equation}\label{biasul}
V:=v_-\Pi_- + v_+\Pi_+.
\end{equation}

We will now consider the Hamiltonian describing the evolution under
an adiabatic introduction of the electric bias. Consider a smooth switch-on
function $\chi$ which fulfills 
$\chi,\chi'\in L^1((-\infty,0))$, $\chi(0)=1$ (note that the first two
conditions imply $\lim_{t\to -\infty}\chi(t)=0$). Let $\eta>0$ be the
adiabatic parameter, and define $\chi_\eta(t):=\chi(\eta t)$, 
$V_\eta(t,x):=\chi_\eta(t)V(x)$. At the one-body level we have
\begin{equation}\label{atreia2}
\mathbf{K}_{L,\eta}(t)\,:=\,\mathbf{H}_L+V_\eta(t,Q)\,=\,\mathbf{H}_L+\chi(\eta t)V(Q)
\end{equation} 
where $V(Q)$ denotes the bounded self-adjoint operator of
multiplication with the step 
function $V(x)$. We shall use the notation 
$\mathbf{K}_{L,\eta}:=\mathbf{K}_{L,\eta}(0)=\mathbf{H}_L+V(Q)$.
The evolution defined by this time-dependent Hamiltonian is described
by a unitary propagator $W(t,s)$, strong solution on the domain of
$\mathbf{H}_L$ of the following Cauchy problem:
\begin{equation}\label{W}
\left\{
 \begin{array}{l}
  -i\partial_tW_{L,\eta}(t,s)=-\mathbf{K}_{L,\eta}(t)W_{L,\eta}(t,s)\\
  W_{L,\eta}(s,s)=1
 \end{array}
\right.
\end{equation} 
for $(s,t)\in\mathbb{R}^2$.

\begin{remark}\label{domeniuinvariant}
 For any $L<\infty$ and $\eta>0$, the operators 
$\{\mathbf{K}_{L,\eta}(t)\}_{t\in\mathbb{R}}$ are self-adjoint in
$\mathcal{H}^L$, having a common domain equal to 
$\mathbb{H}_D(\mathcal{X}_L)$. They are strongly differentiable with respect to
$t\in\mathbb{R}$ with a bounded self-adjoint derivative 
$$
\partial_t\mathbf{K}_{L,\eta}(t)\,=\,\eta\,\chi(\eta t)V(Q).
$$
\end{remark}

The following result is standard and we state it without proof:
\begin{proposition}
 The problem \eqref{W} has a unique solution which is unitary and leaves
 the domain $\mathbb{H}_D(\mathcal{X}_L)$ invariant for any pair
 $(s,t)\in\mathbb{R}^2$. For any triple $(s,r,t)\in\mathbb{R}^3$ it 
 satisfies the relation
 $W_{L,\eta}(t,r)W_{L,\eta}(r,s)=W_{L,\eta}(t,s)$. Moreover it also satisfies the equation 
\begin{equation}\label{apatra2}
-i\partial_sW_{L,\eta}(t,s)=W_{L,\eta}(t,s)\mathbf{K}_{L,\eta}(s).
\end{equation}
\end{proposition}

Later on we will need the following simple but important result:
\begin{proposition}\label{est-HW}
 Let us define the commutator
 $$[\mathbf{H}_L,W_{L,\eta}(t,s)]=\mathbf{H}_LW_{L,\eta}(t,s)-W_{L,\eta}(t,s)\mathbf{H}_L$$
on the dense domain $\mathbb{H}_D(\mathcal{X}_L)$ (left invariant by the
evolution). Then the commutator can be extended by continuity to a 
bounded operator on $\mathcal{H}^L$ such that: 
\begin{align}\label{Omega-inv}
& \sup_{(s,t)\in\mathbb{R}_-^2}\left \Vert
  [\mathbf{H}_L,W_{L,\eta}(t,s)]\right \Vert <\infty, \\
\label{apaispea2}& \sup_{(s,t)\in\mathbb{R}_-^2}
\left \Vert (\mathbf{H}_L+1)W_{L,\eta}(t,s)(\mathbf{H}_L+1)^{-1}\right \Vert <\infty.
\end{align}
\end{proposition}

\begin{proof}
 Let us fix some $\phi$ in $\mathbb{H}_D(\mathcal{X}_L)$, and for
 $t\leq 0$ let us consider the application
\begin{align}\label{asaptea2}
&(-\infty,t)\ni s\mapsto
\phi(s):=[\mathbf{K}_{L,\eta}(s),W_{L,\eta}(s,t)]\phi\in
L^2(\mathcal{X}_L), \nonumber \\
& \mathbf{K}_{L,\eta}(s)\phi(s)\subset \mathbb{H}_D (\mathcal{X}_L)^*,
\end{align}
with $\phi(t)=0$. Due to these properties, the function
$\phi(s)$ is differentiable in the weak topology on 
$\mathbb{H}_D(\mathcal{X}_L)^*$ and we can compute its derivative:
\begin{align}\label{asasea2}
\partial_s\phi(s)&=[\partial_s\mathbf{K}_{L,\eta}(s),W_{L,\eta}(s,t)]\phi
+ 
[\mathbf{K}_{L,\eta}(s),\partial_s W_{L,\eta}(s,t)]\phi \nonumber \\
&=\eta\chi^\prime(\eta s)[V(Q),W_{L,\eta}(s,t)]\phi -i
\mathbf{K}_{L,\eta}(s)
[\mathbf{K}_{L,\eta}(s),W_{L,\eta}(s,t)]\phi\nonumber \\
&=-i \mathbf{K}_{L,\eta}(s)\phi(s)+\psi(s),
\end{align}
where $\psi(s):=\eta\chi\prime(\eta
s)[V(Q),W_{L,\eta}(s,t)]\phi$. Thus, 
taking into account the above solution of the Cauchy problem
(\ref{W}), 
we can write the solution for $\phi(s)$ as
$$
\phi(s)=\int_t^s\,W_{L,\eta}(s,r)\psi(r)\,dr.
$$
Let us recall that $W_{L,\eta}(t,r)$ is unitary, the function
$\chi^\prime$ belongs to 
$L^1(\mathbb{R}_-)$ and $\|V(Q)\|\leq\max\{v_-,v_+\}$ so that
$$
\|\phi(s)\|_{\mathcal{H}^L}\leq\|\chi^\prime\|_{L^1(\mathbb{R}_-)}\max\{v_-,v_+\}\|\phi\|,\qquad
 \forall s\in(-\infty,t), \forall t\in\mathbb{R}_-.
$$
Thus for any pair of functions $\psi$ and $\phi$ in
$\mathbb{H}_D(\mathcal{X}_L)$, 
dense in $\mathcal{H}^L$, we have
$$
|\langle \psi, [\mathbf{K}_{L,\eta}(s),W_{L,\eta}(s,t)]\phi\rangle |
\,\leq\,\|\chi^\prime\|_{L^1(\mathbb{R}_-)}\max\{v_-,v_+\}\|\psi\|\|\phi\|.
$$
After a straightforward density argument, the use of Riesz' representations
theorem, and the fact that $V(Q)$ is bounded, 
we obtain a uniform estimate as in \eqref{Omega-inv} if 
$-\infty<s\leq t\leq 0$. To obtain the same for $-\infty<t\leq s \leq
0$ we remind that
$W_{L,\eta}(s,t)=W_{L,\eta}(t,s)^*$ and for 
$\psi$ and $\phi$ in $\mathbb{H}_D(\mathcal{X}_L)$ we can write (using
the invariance of the domain $\mathbb{H}_D(\mathcal{X}_L)$ under $W_{L,\eta}(t,s)$)
$$
-\langle [\mathbf{H}_L,W_{L,\eta}(s,t)]\phi,\psi\rangle 
=\langle \phi,[\mathbf{H}_L,W_{L,\eta}(t,s)]\psi \rangle,
$$
and we easily get the desired conclusion. Finally, \eqref{apaispea2}
is an easy consequence of \eqref{Omega-inv}. 
\end{proof}

\vspace{0.5cm}

Now let us consider the second quantization of the above time-dependent evolution.
$$
d\Gamma(\mathbf{K}_{L,\eta}(t))=\boldsymbol{\mathfrak{h}}_L+\chi_\eta(t)d\Gamma(V(Q)).
$$
Let us first observe that the perturbation $d\Gamma(V(Q))$ is no 
longer bounded on $\mathfrak{F}_L$. It is rather easy to verify that
the family $\{\mathfrak{w}_{L,\eta}(t,s)\}_{(s,t)\in\mathbb{R}^2}$,
that is well defined by $\mathfrak{w}_{L,\eta}(t,s):=\Gamma(W_{L,\eta}(t,s))$ 
gives the unique (and unitary) solution of the Cauchy problem:
\begin{equation}\label{aopta2}
 \left\{
 \begin{array}{l}
  -i\partial_t\mathfrak{w}_{L,\eta}(t,s)=-d\Gamma\left(\mathbf{K}_{L,\eta}(t)\right)
\mathfrak{w}_{L,\eta}(t,s)\\
  \mathfrak{w}_{L,\eta}(s,s) = 1.
 \end{array}
 \right.
\end{equation}

\subsection{The nonequilibrium state}

In our adiabatic approach, the nonequilibrium state is completely
characterized by a density matrix which solves the Liouville
equation. Its initial condition at $t=-\infty$ is the Gibbs
equilibrium state $\mathfrak{p}_L$, associated to the Hamiltonian 
$\boldsymbol{\mathfrak{h}}_L$ at temperature $\beta^{-1}$ and 
chemical potential $\mu$. We let the state evolve with the adiabatic
time-dependent Hamiltonian $\mathbf{K}_{L,\eta}(t)$ up to $t=0$. Then
we perform the {{} thermodynamic} limit $L\rightarrow\infty$, and finally
we should take the adiabatic limit $\eta\rightarrow 0$. {{} This last step
raises} for the moment some important technical difficulties. We
will avoid them by only considering the \textit{linear response}
behavior, i.e. the linear contribution in the bias $v_+-v_-$. The
proof of the existence of the {{} thermodynamic and }adiabatic limit for the full density
matrix is a challenging open problem.

Now let us formulate and solve the Liouville equation. We first want 
to find a weak solution to the equation:
\begin{equation}\label{anoua2}
 \left\{
 \begin{array}{l}
  i\partial_t
  \mathfrak{p}_{L,\eta}(t)=[d\Gamma\left(\mathbf{K}_{L,\eta}(t)\right),\mathfrak{p}_{L,\eta}(t)], \quad t<0, \\
\lim_{t\to -\infty}\mathfrak{p}_{L,\eta}(t)=\mathfrak{p}_{L}.
 \end{array}
 \right.
\end{equation} 
Let us note that for a fixed $s$, the operator (see also \eqref{aopta2})
$$\mathfrak{w}_{L,\eta}(t,s)\mathfrak{p}_L\mathfrak{w}_{L,\eta}(t,s)^*$$
solves the differential equation, but does not obey the initial condition. What
we do next is to take care of this.  

The state $\mathfrak{p}_L$ being defined in
(\ref{G-state}) commutes with any function of
$\boldsymbol{\mathfrak{h}}_L$. Thus we can write:
\begin{align}\label{azecea2}
&\mathfrak{w}_{L,\eta}(t,s)\mathfrak{p}_L\mathfrak{w}_{L,\eta}(t,s)^*=\mathfrak{w}_{L,\eta}(t,s)e^{i(t-s)
\boldsymbol{\mathfrak{h}}_L}\mathfrak{p}_Le^{-i(t-s)\boldsymbol{\mathfrak{h}}_L}\mathfrak{w}_{L,\eta}(t,s)^*\nonumber
 \\
&=\Gamma\left(W_{L,\eta}(t,s)e^{i(t-s)\mathbf{H}_L}\right)\mathfrak{p}_L
\Gamma\left(W_{L,\eta}(t,s)e^{-i(t-s)\mathbf{H}_L}\right)^*,
\end{align}
and we can consider the two-parameter family of unitaries:
\begin{equation}\label{Omega-L-eta}
 \Omega_{L,\eta}(t,s)\,:=\,W_{L,\eta}(t,s)e^{i(t-s)\mathbf{H}_L},\quad (s,t)\in\mathbb{R}^2.
\end{equation}
We will now show that the solution at time $t$ is given by the
following strong limit (norm limit on each sector of fixed number of particles):
\begin{equation}\label{aunspea2}
 \mathfrak{p}_{L,\eta}(t)\,=\,\underset{s\rightarrow-\infty}{\lim}\Gamma(\Omega_{L,\eta}(t,s))
\mathfrak{p}_L\Gamma(\Omega_{L,\eta}(t,s))^*.
\end{equation}
Let us first show that the limit exists. 
Due to the continuity of the application $\Gamma$, it is enough to
prove the existence of the norm limit 
$\underset{s\rightarrow-\infty}{\lim}\Omega_{L,\eta}(t,s)$ in the one particle sector. For that
we use the following differential equality, valid in strong sense on
the domain of $\mathbf{H}_L$:
\begin{align}\label{adoispea2}
 -i\partial_s\Omega_{L,\eta}(t,s)=W_{L,\eta}(t,s)(\mathbf{K}_{L,\eta}(s)-\mathbf{H}_L)e^{i(t-s)\mathbf{H}_L}.
\end{align} 
If we denote $\tilde{V}(r):=e^{ir\mathbf{H}_L}V(Q)e^{-ir\mathbf{H}_L}$ we observe that 
\begin{equation}\label{atreispea2}
\left\{
\begin{array}{l}
 -i\partial_s\Omega_{L,\eta}(t,s)=\chi(\eta s)\Omega_{L,\eta}(t,s)\tilde{V}(s-t)\\
\Omega(s,s)=1,
\end{array}
\right.
\end{equation} 
and thus
\begin{equation}\label{ec-Omega}
 \Omega_{L,\eta}(t,s)=1\,-\,i\int_s^t\chi(\eta r)\Omega_{L,\eta}(t,r)\tilde{V}(r-t)dr.
\end{equation} 
Since $\Omega_{L,\eta}(t,r)$ is unitary for any pair
$(s,t)\in\mathbb{R}_-^2$, 
$\|\tilde{V}(r-t)\|$ is uniformly bounded in $r$, and $\chi$ is integrable, we conclude that the limit 
\begin{equation}\label{aoptuspea2}
 \underset{s\rightarrow-\infty}{\lim}\Omega_{L,\eta}(t,s)\,=:\,\Omega_{L,\eta}(t)
\end{equation} 
exists in the norm topology of $\mathbb{B}[\mathcal{H}^L]$ and thus
defines a unitary operator. 
\begin{proposition}\label{est-Omega}
 The operator $\Omega_{L,\eta}(t)$ preserves the domain of
 $\mathbf{H}_{L}$, and the densely defined commutator
 $[\mathbf{H}_{L},\Omega_{L,\eta}(t)]$ can be extended to a bounded
 operator on $\mathcal{H}^L$, with a norm which is uniformly bounded in $t\leq 0$. 
\end{proposition}

\begin{proof} If $\psi$ belongs to the domain of $\mathbf{H}_L$, then
  for every $\phi\in {{} \mathcal{H}^L}$ we have: 
\begin{align}\label{asaispea2}
&|\langle
\mathbf{H}_L\psi,\Omega_{L,\eta}(t)(\mathbf{H}_L+1)^{-1}\phi\rangle|\nonumber \\
&=
\lim_{s\to -\infty}|\langle
\mathbf{H}_L\psi,W_{L,\eta}(t,s)(\mathbf{H}_L+1)^{-1}e^{i(t-s)\mathbf{H}_L}\phi\rangle|
\leq C||\psi||\; ||\phi||,
\end{align}
where we used \eqref{Omega-L-eta}, \eqref{aoptuspea2} and
\eqref{apaispea2}. Now since $\mathbf{H}_L$ is self-adjoint, it means
that $\Omega_{L,\eta}(t)(\mathbf{H}_L+1)^{-1}\phi$ belongs to the
domain of $\mathbf{H}_L$, hence:
\begin{equation}\label{anouaspea2}
\sup_{t\leq 0}\Vert
\mathbf{H}_L\Omega_{L,\eta}(t)(\mathbf{H}_L+1)^{-1}\Vert \leq C. 
\end{equation}
Now for $\psi$ and $\phi$ in the domain of $\mathbf{H}_L$ we can write
\begin{align}\label{prima3}
&|\langle \psi,
[\mathbf{H}_L,\Omega_{L,\eta}(t)]\phi\rangle|\nonumber \\
&=
\lim_{s\to -\infty}|\langle\psi,
[\mathbf{H}_L,W_{L,\eta}(t,s)]e^{i(t-s)\mathbf{H}_L}\phi\rangle|
\leq C||\psi||\; ||\phi||,
\end{align}
where we used \eqref{Omega-inv}. A density argument finishes the proof.
\end{proof}

\vspace{0.5cm}

Up to now we have shown that $\mathfrak{p}_{L,\eta}$ is a weak
solution to the Liouville equation \eqref{anoua2}. This density matrix
is a trace class operator. The key identity which gives the
expectation of a one-body bounded observable
$\mathbf{T}\in\mathbb{B}(\mathcal{H}^L)$ lifted to the Fock space is
(see also \eqref{prima2}):
\begin{align}\label{adoua3}
\Tr_{\mathfrak{F}_L}\left(\mathfrak{p}_{L,\eta}(t)d\Gamma(\mathbf{T})\right)&=
\Tr_{\mathfrak{F}_L}\left(\mathfrak{p}_L\Gamma(\Omega_{L,\eta}(t))^*d\Gamma(\mathbf{T})
\Gamma(\Omega_{L,\eta}(t))\right)\nonumber \\
&=\Tr_{\mathfrak{F}_L}\left(\mathfrak{p}_L
  d\Gamma\{\Omega_{L,\eta}^*(t)\mathbf{T}\Omega_{L,\eta}(t)\}\right )\nonumber \\
&=\Tr_{\mathcal{H}^L}\left(\Omega_{L,\eta}(t)\rho(\mathbf{H}_L)
\Omega_{L,\eta}^*(t)\mathbf{T} \right).
\end{align}
We stress that all this works because $\mathcal{X}_L$ is finite. The
main conclusion is that if we are only interested in expectations of
one body observables, the effective one-particle density matrix is:
\begin{equation}\label{atreia3}
\rho_{L,\eta}(t):=\Omega_{L,\eta}(t)\rho(\mathbf{H}_L)\Omega_{L,\eta}^*(t)\in\mathbb{B}_1(\mathcal{H}^L).
\end{equation}
We can now prove that the above mapping is differentiable with respect
to $t$ in the trace norm topology. We write:
\begin{align}\label{apatra3}
&\rho_{L,\eta}(t)=
\underset{s\rightarrow-\infty}{\lim}\Omega_{L,\eta}(t,s)\rho(\mathbf{H}_L)\Omega_{L,\eta}^*(t,s)
=\,\underset{s\rightarrow-\infty}{\lim}W_{L,\eta}(t,s)\rho(\mathbf{H}_L)W_{L,\eta}^*(t,s)\nonumber \\ 
&=W_{L,\eta}(t,0)\left\{\underset{s\rightarrow-\infty}{\lim}W_{L,\eta}(0,s)e^{-is\mathbf{H}_L}
\rho(\mathbf{H}_L)e^{is\mathbf{H}_L}W_{L,\eta}^*(0,s)\right\}W_{L,\eta}(t,0)^*\nonumber \\
&=W_{L,\eta}(t,0)\rho_{L,\eta}(0)W_{L,\eta}(t,0)^*=W_{L,\eta}(t,0)\Omega_{L,\eta}(0)\rho(\mathbf{H}_L)
\Omega_{L,\eta}^*(0)W_{L,\eta}^*(t,0).
\end{align}
We have
$\mathbf{H}_L\rho(\mathbf{H}_L)\mathbf{H}_L\in\mathbb{B}_1[\mathcal{H}^L]$
because of the exponential decay of $\rho$. Using \eqref{anouaspea2}
we obtain that $\mathbf{H}_L\rho_{L,\eta}(0)\mathbf{H}_L$ is also
trace class. Then using \eqref{W} and \eqref{apaispea2} we conclude
that $t\mapsto \rho_{L,\eta}(t)\in \mathbb{B}_1[\mathcal{H}^L]$ is
differentiable and:
\begin{equation}\label{acincea3}
\left.\partial_t\,\rho_{L,\eta}(t)\,\right|_{t=0}\,=\,-i\,\left[\mathbf{K}_{L,\eta},\rho_{L,\eta}\right]
\,\in\,\mathbb{B}_1[\mathcal{H}^L].
\end{equation}

As we have said at the {{} beginning} of this subsection, studying the
limits $L\rightarrow\infty$ and $\eta\rightarrow 0$ on the above
formula for the whole state {{} seems to be}  rather difficult, and we will only
consider the first order correction with respect to the potential
bias, obtained by considering the equation (\ref{ec-Omega}) in the limit $s\rightarrow -\infty$:
\begin{align}\label{asasea3}
\Omega_{L,\eta}(t)&=1\,-\,i\int_{-\infty}^t\,\chi(\eta
r)\Omega_{L,\eta}(t,r) 
\tilde{V}(r-t)\,dr\, \nonumber \\
&\sim\,1\,-\,i\int_{-\infty}^t\,\chi(\eta r)\tilde{V}(r-t)\,dr\,+\mathcal{O}(V^2).
\end{align}
{\it Let us point out here that a control of the above rest
  $\mathcal{O}(V^2)$ is a difficult task that we will not consider for the moment.}

Having in mind the above argument, we define as the {\it linear response state at time 0}:
\begin{equation}\label{LR-state}
\tilde{\rho}_{L,\eta}:=\rho(\mathbf{H}_L)-\left[\mathcal{V}_\eta,\rho(\mathbf{H}_L)\right],
\end{equation} 
where:
\begin{equation}\label{LR-Omega}
\mathcal{V}_\eta:=i\int_{-\infty}^0\,\chi(\eta r)\tilde{V}(r)\,dr.
\end{equation}

\subsection{The current}

The main advantage of our approach is that we can define the current
coming out of a lead as the time derivative of its charge. We define
the {\it charge operators} at finite volume, to be the second
quantization of projections $\Pi_\pm$ (see \eqref{adoua2}):
\begin{equation}\label{asaptea3}
\mathfrak{Q}_\pm:=d\Gamma(\Pi_\pm).
\end{equation}
The {\it average charge at time} $t$ is given by:
\begin{equation}\label{aopta3}
q(t):=\Tr_{\mathfrak{F}_L}\left(\mathfrak{p}_{L,\eta}(t)\mathfrak{Q}_\pm\right)=
\Tr_{\mathcal{H}^L}\left(\rho_{L,\eta}(t)\Pi_\pm\right).
\end{equation}
By differentiating with respect to $t$ and using the conclusion of the
previous subsection we obtain the {\it average {{} current} at time} $t=0$:
\begin{align}\label{anoua3}
j_{L,\eta}&=-i\Tr_{\mathcal{H}^L}\left(\left[\mathbf{K}_{L,\eta},\rho_{L,\eta}\right]\,\Pi_\pm\right)\nonumber \\
&=-i\Tr_{\mathcal{H}^L}\left(\left[\mathbf{H}_{L},\rho_{L,\eta}\right]\,\Pi_\pm\right)-
i\Tr_{\mathcal{H}^L}\left(\left[V(Q),\rho_{L,\eta}\right]\,\Pi_\pm\right).
\end{align}
But (see \eqref{biasul}) 
$\Tr_{\mathcal{H}^L}\left(\left[V(Q),\rho_{L,\eta}\right]\,\Pi_\pm\right)={{}-}
\Tr_{\mathcal{H}^L}\left(\rho_{L,\eta}\left[V(Q),\,\Pi_\pm\right]\right)=0$
and we deduce that:
\begin{align}\label{azecea3}
j_{L,\eta}&
=-i\Tr_{\mathcal{H}^L}\left(\left[\mathbf{H}_{L},\rho_{L,\eta}\right]\,\Pi_\pm\right)\nonumber \\
&=-i\Tr_{\mathcal{H}^L}\left((\mathbf{H}_{L}+1)\rho_{L,\eta}
  (\mathbf{H}_{L}+1)\mathbf{R}_L\Pi_\pm\,-\,
\mathbf{R}_L(\mathbf{H}_{L}+1)\rho_{L,\eta}
(\mathbf{H}_{L}+1)\Pi_\pm\right)\nonumber \\
&=-i\Tr_{\mathcal{H}^L}\left((\mathbf{H}_{L}+1)\rho_{L,\eta}
  (\mathbf{H}_{L}+1)\,\left[\mathbf{R}_L,\Pi_\pm\right]\right).
\end{align}
Because we only are interested in the {\it linear response}, we will use
(\ref{LR-state}). Remember that
$\rho(\mathbf{H}_L)\mathbf{H}_L^2\in\mathbb{B}_1[\mathcal{H}^L]$. An
important observation is that the following commutator defined as a 
sesquilinear form on $\mathbb{H}_D(\mathcal{X}_L)^2$, can be extended
to a bounded operator on $L^2(\mathcal{X}_L)$ since:
\begin{align}\label{aunspea3}
\left[\mathbf{H}_L
  ,\mathcal{V}_\eta\right]\,&=\,\int_{-\infty}^0\,\chi(\eta
s)\,\left\{\partial_s\,e^{is\mathbf{H}_L}Ve^{-is\mathbf{H}_L}\right\}\,ds\,\nonumber
\\
&=\,V-\eta\int_{-\infty}^0\,\chi^\prime(\eta s)e^{is\mathbf{H}_L}Ve^{-is\mathbf{H}_L}\,ds.
\end{align}
Note that we do not have to commute $\mathbf{H}_L$ with $V$ in order
to get this result. In fact $[\mathbf{H}_L,V]$ is quite singular due
to the sharp characteristic functions from the definition of $V$. 

Thus the second term in (\ref{LR-state}) is also trace-class and we
can write the 
{\it linear response average {{} current}} at time $t=0$ as:
\begin{align}\label{LR-curent}
 \tilde{j}_{L,\eta}\,&:=\,-i\Tr_{\mathcal{H}^L}\left((\mathbf{H}_{L}+1)\tilde{\rho}_{L,\eta} 
(\mathbf{H}_{L}+1)\,\left[\mathbf{R}_L,\Pi_\pm\right]\right)\nonumber \\
&=\,-i\Tr_{\mathcal{H}^L}\left((\mathbf{H}_{L}+1)\rho(\mathbf{H}_L)(\mathbf{H}_{L}+1)\,
\left[\mathbf{R}_L,\Pi_\pm\right]\right)\nonumber \\
&+i\,\Tr_{\mathcal{H}^L}\left((\mathbf{H}_{L}+1)\left[\mathcal{V}_\eta,\rho(\mathbf{H}_L)\right] 
(\mathbf{H}_{L}+1)\,\left[\mathbf{R}_L,\Pi_\pm\right]\right).
\end{align}
The first term of the last sum is zero due to trace {{} cyclicity}. We then
obtain:
\begin{align}\label{curent}
\tilde{j}_{L,\eta}\,&=\,i\Tr_{\mathcal{H}^L}\left((\mathbf{H}_{L}+1)
\left[\mathcal{V}_\eta,\rho(\mathbf{H}_L)\right]
(\mathbf{H}_{L}+1)\,\left
[\mathbf{R}_L,\Pi_\pm\right]\right)\nonumber \\
&=\,i\Tr_{\mathcal{H}^L}\left((\mathbf{H}_{L}+1)\left[\mathcal{V}_\eta,\mathbf{R}_L\right] 
(\mathbf{H}_{L}+1)\,\left[\rho(\mathbf{H}_L),\Pi_\pm\right]\right),
\end{align} 
where the second equality is obtained by carefully developing the
commutators and using the {{} cyclicity} property of the trace and the fact
that $\mathbf{H}_L^2\rho(\mathbf{H}_L)$ is trace class.

Now we rewrite (\ref{curent}) in a more suitable form for taking the
limits $L\rightarrow\infty$ and $\eta\rightarrow 0$.

We begin by computing the first commutator in (\ref{curent}):
\begin{align}\label{adoispea3}
\left[\mathcal{V}_\eta,\mathbf{R}_L\right] &=\,i\int_{-\infty}^0\,\chi(\eta
r)\left[ \tilde{V}(r),\mathbf{R}_L\right]\,dr\nonumber \\
&=\,i\int_{-\infty}^0\,\chi(\eta r)e^{ir\mathbf{H}_L}\left[ V(Q),\mathbf{R}_L\right]e^{-ir\mathbf{H}_L}\,dr.
\end{align}
Now let us observe that in the strong topology on $\mathbb{B}[\mathcal{H}_L]$ we can write:
$$
\partial_s\,e^{is\mathbf{H}_L} \mathbf{R}_L(z) V(Q) \mathbf{R}_L(z)
e^{-is\mathbf{H}_L}\,=\,i\,e^{is\mathbf{H}_L} 
\left[V(Q) , \mathbf{R}_L(z)\right] e^{-is\mathbf{H}_L},
$$
so that we get:
\begin{align}\label{atreispea3}
\left[\mathcal{V}_\eta,\mathbf{R}_L\right]\,&=\,\int_{-\infty}^0\,\chi(\eta
s)\,\left\{\partial_s\,e^{is\mathbf{H}_L} \mathbf{R}_L V(Q) \mathbf{R}_L
  e^{-is\mathbf{H}_L}\right\}\,ds\nonumber \\
&=\,\mathbf{R}_L V(Q) \mathbf{R}_L\,
-\,\eta\,\int_{-\infty}^0\,\chi^\prime(\eta s)\,e^{is\mathbf{H}_L}
\mathbf{R}_L V(Q) 
\mathbf{R}_L e^{-is\mathbf{H}_L}\,ds.
\end{align}
Inserting the first term in (\ref{curent}), we observe that it gives zero:
\begin{equation}\label{RR}
\Tr_{\mathcal{H}^L}\left(V(Q)\,\left[\rho(\mathbf{H}_L),\Pi_\pm\right]\right)=0,
\end{equation}
due to the trace {{} cyclicity} and the fact that $V(Q)=v_+\Pi_++v_-\Pi_-$
{{} commutes} with $\Pi_\pm$. Thus the final expression is:
\begin{equation}\label{curent-2}
 \tilde{j}_{L,\eta}\,=\,-i\eta\,\int_{-\infty}^0\,\chi^\prime(\eta
 s)\,\Tr_{\mathcal{H}^L}
\left(e^{is\mathbf{H}_L} V(Q) e^{-is\mathbf{H}_L}\left[\rho(\mathbf{H}_L),\Pi_\pm\right]\right)ds.
\end{equation} 

\section{The thermodynamic limit}

\subsection{The limit of the dynamics}

For different values of $L$, the Hamiltonians $\mathbf{H}_L$ are
densely defined in different Hilbert spaces
$\mathcal{H}^L:=L^2(\mathcal{X}_L)$. In order to study the behavior of our system 
at $L\rightarrow\infty$, we embed $\mathcal{H}^L$ in the unique Hilbert space 
$\mathcal{H}:=L^2(\mathcal{X}_\infty)$ (i.e. with infinitely long
leads), using the natural decomposition: 
\begin{equation}\label{last1}
\mathcal{H}=\mathcal{H}^L \oplus L^2(\mathcal{X}_\infty\setminus \mathcal{X}_L).
\end{equation}
We denote by $\Pi_L$ the orthogonal projection corresponding to the
restriction to $\mathcal{X}_L$. Its orthogonal $\Pi_L^\bot$
corresponds to the restriction to $\mathcal{X}_\infty\setminus
\mathcal{X}_L$.  Each operator
$\mathbf{H}_L$ is self-adjoint on the domain $\mathbb{H}_D(\mathcal{X}_L)$ (see \eqref{Hamilt}), while the
resolvent $\mathbf{R}_L(z)$ is a bounded operator. We extend $\mathbf{R}_L(z)$ to $\mathcal{H}$ by the 
natural formula $ \Pi_L \mathbf{R}_L(z)\Pi_L$; this new $z$-dependent operator family is a pseudoresolvent, 
which vanishes on $\Pi_L^\bot\mathcal{H}$. 

Let us denote by 
\begin{equation}\label{atreia44}
\mathbf{H}:=-\Delta+w(Q)
\end{equation}
with $-\Delta$ minus the
usual Laplacian with Dirichlet boundary conditions defined on the
Sobolev space $H^2(\mathcal{X}_\infty)$. We denote by $\mathbf{R}(z)$
the resolvent $(\mathbf{H}-z)^{-1}$ (and with
$\mathbf{R}:=\mathbf{R}(-1)$).

\begin{proposition}\label{s-conv-L}
The sequence of operators $\{\Pi_L\mathbf{R}_L\Pi_L\}_{L>1}$ converges strongly to $\mathbf{R}$.
\end{proposition}
\begin{proof}
Fix $u\in H^2(\mathcal{X}_\infty)$. If $\psi\in C_0^\infty(\R)$,
define $\psi_3(\x):=\psi(x_1)$. If $L$ is large enough such that
${\rm supp}(\psi_3 u)\subset \mathcal{X}_L$, then $\psi_3 u\in
\mathbb{H}_D(\mathcal{X}_L) $. If we denote by $D_1$ the subset of
$H^2(\mathcal{X}_\infty)$ with compact support in the $x_1$ variable,
then $\mathbf{H}$ is essentially self adjoint on it. Then the subspace 
$\mathcal{E}:=(\mathbf{H}+1)D_1$ is dense in the Hilbert space
$\mathcal{H}$. For any $u\in\mathcal{E}$, there exists $L_u>0$ such
that for every $L\geq L_u$ we have $\Pi_L\mathbf{R}u=\mathbf{R}u$, 
$\mathbf{R}u\in\mathbb{H}_D(\mathcal{X}_L)$ and
$\mathbf{H}_L\mathbf{R}u=\mathbf{H}\mathbf{R}u$, and thus we have:
$$
\Pi_L\mathbf{R}_L\Pi_Lu-\mathbf{R}u\,=\,\Pi_L\mathbf{R}_L\left(\Pi_L\mathbf{H}-\mathbf{H}_L\right)\mathbf{R}u=
\Pi_L\mathbf{R}_L\Pi_L^\bot u.
$$
As $\|\mathbf{R}\|\leq1$ and $\|\mathbf{R}_L\|\leq1$ for any $L>0$,
the proof follows after a density argument.
\end{proof}
\begin{corollary}\label{s-conv-H}
The sequence of operators $\{\Pi_Le^{it\mathbf{H}_L}\Pi_L \}_{L>a}$ converges in the
strong topology to $e^{it\mathbf{H}}$, uniformly for $t$ in any compact subset of $\mathbb{R}$.
\end{corollary}
\begin{proof}
Let us fix $u\in\mathcal{H}$ and some $t\in\mathbb{R}$. We can take $u$ with compact support, such that 
$\Pi_Lu=u$ for $L$ large enough. For any
continuous function which vanishes at infinity $\psi\in C_\infty(\mathbb{R})$ consider 
$\psi(\mathbf{H}_L)$ defined on $\mathcal{H}^L$
through functional calculus. We 
define the operator $\tilde{\psi}(\mathbf{H}_L)=\Pi_L\psi(\mathbf{H}_L)\Pi_L$. 
Let us observe that $C_\infty(\mathbb{R})$ is the
norm closure of the algebra defined by the 'resolvent functions' 
$\{\mathfrak{r}_z(t):=(t-z)^{-1}\}_{z\in\mathbb{C}\setminus\mathbb{R}}$
and our definition of $\tilde{\psi}(\mathbf{H}_L)$ agrees with the
definition of the extension of the resolvent $\mathbf{R}_L$. Thus the
strong convergence of the pseudoresolvents for $L\rightarrow\infty$
immediately implies the strong convergence of the operators
$\tilde{\psi}(\mathbf{H}_L)$ to $\psi(\mathbf{H})$ defined by the
usual functional calculus on $\mathcal{H}$. Now let us choose $\psi\in
C_\infty(\mathbb{R})$ such that $\psi(\lambda)=1$ for $|\lambda|\leq 1$ and
$\psi(\lambda)=0$ for 
$|\lambda|\geq 2$. For any $E>0$ let us denote by
$\psi_E(\lambda):=\psi(E^{-1}\lambda)$. Using the above arguments, let
us fix $\epsilon>0$ as small as we want and choose $E>0$ such that 
$\|(1-\psi_E(\mathbf{H}))u\|\leq\epsilon/2$. Then there exists
$L_\epsilon$ large enough such that 
$\|(1-\widetilde{\psi_E}(\mathbf{H}_L))u\|\leq \epsilon$ for any
$L\geq L_\epsilon$. Then 
\begin{align*}
&\|(\Pi_Le^{it\mathbf{H}_L}\Pi_L-e^{it\mathbf{H}})u\|\nonumber \\
&\leq \|(\widetilde{e^{it\cdot}{\psi_E}}(\mathbf{H}_L)-\widetilde{e^{it\cdot }
\psi_E}(\mathbf{H}))u\|+\|(1-\widetilde{\psi_E}(\mathbf{H}_L))u\|+\|(1-\psi_E(\mathbf{H}))u\|.
\end{align*}
For the first term we observe that for any fixed $t\in\mathbb{R}$
the function 
$\lambda\mapsto\psi(\lambda)e^{it\lambda}$ belongs to
$C_0(\mathbb{R})$ and 
thus we can use once again the strong resolvent convergence in order
to control it. Due to the continuity of the map $t\mapsto
e^{it\lambda}$ (at fixed $\lambda$), 
we deduce that we have the strong convergence {{} locally} uniformly in $t\in\mathbb{R}$.
\end{proof}

\subsection{The limit of the current}

 Let us fix some small enough $r>0$ and consider the positively oriented contour: 
\begin{equation}\label{conturint}
\mathcal{C}_r=\left\{\lambda+ir\mid \lambda\in\mathbb{R}_+\right\}\cup\left\{re^{i\theta}
\mid\theta\in[\pi/2,3\pi/2]\right\}\cup\left\{\lambda-ir\mid\lambda\in\mathbb{R}_+\right\}
\subset\mathbb{C}.
\end{equation}
Then by {{} analytic} functional calculus we can write 
(due to the {{} analyticity} and the decay properties of the function $\rho$):
\begin{equation}\label{densimatri}
\rho(\mathbf{H}_L)\,=\,\frac{i}{2\pi}\,\int_{\mathcal{C}_r}\,\rho(z)\,\mathbf{R}_L(z)\,dz.
\end{equation}
All these operators can be extended to the whole Hilbert space by the procedure 
$\Pi_L\rho(\mathbf{H}_L)\Pi_L$, (i.e. when considered in $\mathcal{H}$, $\rho(H_L)$ stands for
 $\Pi_L\rho(\mathbf{H}_L)\Pi_L$)  but for notational simplicity we drop the cut-off projectors.
The main result of this subsection is:
\begin{proposition}\label{convintrasa}
The operators $[\rho(\mathbf{H}_L),\Pi_+]$, $1<L\leq \infty$, are
trace class. Moreover, for every $n>1$ there exists $C>0$ such that for $L>1$ we have:
$$\left \Vert
  [\rho(\mathbf{H}),\Pi_+]-[\rho(\mathbf{H}_L),\Pi_+]\right\Vert _{\mathbb{B}_1(\mathcal{H})}\leq
CL^{-n}.$$ 
\end{proposition}
\begin{proof} 
Let $0\leq
\chi(x_1)\leq 1$ be a $C^\infty$ cut-off function such that $\chi=1$ on
$x_1\geq a+2$ and $\chi=0$ on $x_1\leq a+1$. We have $\Pi_+\chi=\chi$, and
the following identity holds:
\begin{align}\label{prima}
[\mathbf{R}_L(z),\Pi_+]=[\mathbf{R}_L(z),\chi]+\mathbf{R}_L(z)\Pi_+
(1-\chi)-\Pi_+ (1-\chi)\mathbf{R}_L(z).
\end{align}
Now choose two functions $\Phi,\tilde{\Phi}\in C^\infty(\mathbb{R}^3)$ such that
$0\leq \Phi,\tilde{\Phi}\leq 1$, $\Phi(\x)=1$ if $|x_1|\leq 1/4$, $\Phi(\x)=0$
if $|x_1|\geq 1/2$, and $\Phi^2(\x)+\tilde{\Phi}^2(\x)=1$ on
$\mathbb{R}^3$. Denote by $\Phi_L(\x):=\Phi(\x/L)$ and 
$\tilde{\Phi}_L(\x):=\tilde{\Phi}(\x/L)$. We have: 
\begin{align}\label{adoua}
&\Phi_L^2(\x)+\tilde{\Phi}_L^2(\x)=1, \; \forall \x\in \mathbb{R}^3, \\
&{\rm supp}(\Phi_L)\subset \{|\x|\leq L/2\},\quad 
{\rm supp}(\tilde{\Phi}_L)\subset \{|\x|\geq L/4\},\label{atreia}\\
&\sup_{\x\in \mathbb{R}^3}\left
  \{|D^\alpha\Phi_L|(\x)+|D^\alpha\tilde{\Phi}_L|(\x)\right \}\leq C(\alpha)L^{-|\alpha|}. 
\end{align}
If
$z\in\mathbb{C}\setminus \mathbb{R}$ define:
\begin{equation}\label{eszet}
S_L(z):=\Phi_L \mathbf{R}(z)\Phi_L+\tilde{\Phi}_L\mathbf{R}_L(z)\tilde{\Phi}_L.
\end{equation}
The range of this operator is included in the domain of $\mathbf{H}_L$
and we can write: 
 \begin{align}\label{pseudoresolventequ}
(\mathbf{H}_L-z)S_L(z)&=1+T_L(z),\\
\label{apatra} T_L(z)=[-2(\nabla \Phi_L)\nabla &
-(\Delta\Phi_L)]\mathbf{R}(z)\Phi_L+
[-2(\nabla \tilde{\Phi}_L)\nabla
-(\Delta\tilde{\Phi}_L)]\mathbf{R}_L(z)\tilde{\Phi}_L.
\end{align}
This leads to a resolvent-like equation:
 \begin{align}\label{acincea}
\mathbf{R}_L(z)=S_L(z)-\mathbf{R}_L(z)T_L(z).
\end{align}
Using standard Combes-Thomas estimates \cite {combes}, \cite{CN}, one can prove
the following lemma (given without proof): 
\begin{lemma}\label{CTestim}
Denote by 
$d(\x) :=\sqrt{1+x_1^2}$. If $z\in \mathcal{C}_r$ (see
\eqref{conturint}) we denote by $\lambda=\sqrt{1+\Re(z)^2}$. Fix
$0\leq \beta_2<\beta_1\leq 1$. Then there exist
$\delta >0$, $p>0$ and $C>0$ such that uniformly in $L>1$ (including
$L=\infty)$) and $z\in \mathcal{C}_r$:
\begin{align}\label{asasea}
&||e^{\pm \frac{\delta d}{\lambda}}D^\alpha\mathbf{R}_L(z)e^{\mp
  \frac{\delta d}{\lambda}}||_{B(L^2)}+||e^{\pm \frac{\delta d}{\lambda}}\mathbf{R}_L(z)e^{\mp
  \frac{\delta d}{\lambda}}||_{B(L^2,L^\infty)}\leq C \lambda^p,\quad |\alpha|\leq 1,\\
&||e^{-\beta_1\frac{\delta
    d}{\lambda}}\mathbf{R}_L(z)e^{\beta_2\frac{\delta d}{\lambda}}||_{HS}\leq
C \lambda^p.\label{asaptea}
\end{align}
\end{lemma}
We are now ready to prove Proposition \ref{convintrasa}. 
The idea is to use \eqref{prima} in \eqref{densimatri} and prove that each term on the
right hand side converge in the trace norm to the operator with 
$L=\infty$. 

Let us first prove that $\rho(\mathbf{H})\Pi_+(1-\chi)\in
\mathbb{B}_1(\mathcal{H})$. Because $\rho$ has an exponential decay at infinity, we
can integrate arbitrary many times by parts in \eqref{densimatri} and replace
$\mathbf{R}(z)$ by $\mathbf{R}^N(z)$, $N\geq 2$: 
\begin{equation}\label{aopta}
\rho(\mathbf{H})\,=\,\frac{i}{2\pi}\,\int_{\mathcal{C}_r}\,\rho_N(z)\,\mathbf{R}^N(z)\,dz,
\end{equation} 
where $\rho_N$ still decays exponentially at infinity. Take $N=2$. Let
us prove that the integrand in \eqref{aopta} becomes trace class if we
multiply it by the localized multiplication operator
$\phi:=\Pi_+(1-\chi)$. Indeed, for $0<\alpha_2<\alpha_1<1$ we have:
\begin{align}\label{anoua}
\mathbf{R}^2(z)\phi=\left \{\mathbf{R}(z)e^{-\alpha_2\frac{\delta
    d}{\lambda}}\right \}\left \{e^{\alpha_2\frac{\delta
    d}{\lambda}}\mathbf{R}(z)e^{-\alpha_1\frac{\delta
    d}{\lambda}}\right \}  e^{\alpha_1\frac{\delta d}{\lambda}}\phi.
\end{align}
The operator $e^{\alpha_1\frac{\delta d}{\lambda}}\phi$ is bounded,
while the other two are Hilbert-Schmidt (see \eqref{asaptea}). Thus:
\begin{align}\label{azecea}
||\mathbf{R}^2(z)\phi||_{\mathbb{B}_1(\mathcal{H})}\leq
C\lambda^p,\quad \forall z\in \mathcal{C}_r.
\end{align}
Now since $\rho_N$ has an exponential decay, the integral with respect
to $z$ defines a trace class operator. 

Let us now prove that for any $n\geq 1$ there exists $C>0$ such that:
\begin{equation}\label{aunspea}
\left \Vert
  \Pi_+(1-\chi)\rho(\mathbf{H})-\Pi_+(1-\chi)\rho(\mathbf{H}_L)\right
\Vert _{\mathbb{B}_1(\mathcal{H})}\leq
CL^{-n}.
\end{equation}
Using $\tilde{\Phi}_L\Pi_+(1-\chi)=0$ and
$\Phi_L\Pi_+(1-\chi)=\Pi_+(1-\chi)$ for $L>L_0$, and introducing equations
\eqref{acincea} and \eqref{eszet} in \eqref{densimatri} we obtain
(denote again by $\phi=\Pi_+(1-\chi)$):
\begin{align}\label{adoispea}
\phi\,\rho(\mathbf{H})-\phi\,\rho(\mathbf{H}_L)=\phi\,\rho(\mathbf{H})(\Phi_L-1)+
\frac{i}{2\pi}\,\int_{\mathcal{C}_r}\,\rho(z)\phi\,\mathbf{R}_L(z)T_L(z) dz.
\end{align}
Both operators on the right hand side are trace class with a fast
decaying trace norm. For the first one we use \eqref{aopta} with $N=2$
and reason as in \eqref{anoua}:
\begin{align}\label{atreispea}
&\phi\,\mathbf{R}^2(z)(1-\Phi_L)\\
&=e^{\alpha_1\frac{\delta d}{\lambda}}\phi 
\left \{e^{-\alpha_1\frac{\delta
    d}{\lambda}}\mathbf{R}(z)e^{\alpha_2\frac{\delta
    d}{\lambda}}\right \} \left \{e^{-\alpha_2\frac{\delta d}{\lambda}}\mathbf{R}(z)e^{\alpha_3\frac{\delta
    d}{\lambda}}\right \}e^{-\alpha_3\frac{\delta
  d}{\lambda}}(1-\Phi_L),\nonumber \\
& \nonumber 0<\alpha_3<\alpha_2<\alpha_1<1.
\end{align}
Again $e^{\alpha_1\frac{\delta d}{\lambda}}\phi$ is bounded, the next
two operators are Hilbert-Schmidt, and the last one has a norm bounded
from above by $e^{-\frac{c L}{\lambda}}$, $c>0$, because of the
support properties of $1-\Phi_L$. Hence 
$$\left \Vert \phi\,\mathbf{R}^2(z)(1-\Phi_L)\right\Vert
_{\mathbb{B}_1(\mathbb{R}^3)}\leq C\lambda^p e^{-\frac{c L}{\lambda}}$$
and after integration with respect to $z$ we obtain a decay as in
\eqref{aunspea} (use the exponential decay of $\rho_N$).

As for the other term in \eqref{adoispea}, we have to integrate by
parts with respect to $z$ in order to obtain products containing three
resolvents. One resolvent is used to bound the momentum operator from
the formula of $T_L$, and the other two for creating two
Hilbert-Schmidt factors. Let us consider one typical term: 
$$\phi\,\mathbf{R}_L(z)\Phi_L'\partial_{x_1}\mathbf{R}_L^2(z).$$
We write: 
\begin{align}\label{apaispea}
&\phi\,\mathbf{R}_L(z)\Phi_L'\partial_{x_1}\mathbf{R}_L^2(z)=\left \{e^{\alpha_1\frac{\delta d}{\lambda}}\phi 
\right \} \left \{e^{-\alpha_1\frac{\delta
    d}{\lambda}}\mathbf{R}_L(z)e^{\alpha_2\frac{\delta
    d}{\lambda}}\right \} \nonumber \\
&\cdot 
\left \{  e^{-\alpha_2\frac{\delta d}{2\lambda}}  \Phi_L'    \right \} 
\left \{e^{-\alpha_2\frac{\delta d}{2\lambda}}\partial_{x_1}\mathbf{R}_L(z)e^{\alpha_2\frac{\delta
    d}{2\lambda}}\right \}e^{-\alpha_2\frac{\delta
  d}{2\lambda}}\mathbf{R}_L(z),\nonumber \\
& 0<\alpha_2<\alpha_1<1.
\end{align}
The first factor in the above product, $e^{\alpha_1\frac{\delta
    d}{\lambda}}\phi $, is bounded. The second and the fifth factors
are Hilbert-Schmidt, while the fourth one is bounded (see again
\eqref{asasea} and \eqref{asaptea}). Now the operator
$e^{-\alpha_2\frac{\delta d}{2\lambda}}  \Phi_L'$ is again bounded by
$e^{-\frac{c L}{\lambda}}$ since $\Phi_L'$ is supported in
$|x_1|>L/4$. Hence the operator in \eqref{apaispea} has an
exponentially small trace norm, and after integration with respect to
$z$ we obtain an estimate as in the right hand side of
\eqref{aunspea}. 

Looking back at \eqref{densimatri} and \eqref{prima} we see that we
also need to prove that $[\rho(\mathbf{H}),\chi]$ is trace class, and  
\begin{equation}\label{acinspea}
\left \Vert [\rho(\mathbf{H}),\chi]-[\rho(\mathbf{H}_L),\chi]
  \right\Vert _{\mathbb{B}_1(\mathcal{H})}\leq CL^{-n}.
\end{equation}
Since (use \eqref{aopta} with $N=2$)
\begin{equation}\label{asaispea}
[\rho(\mathbf{H}),\chi]\,=\,\frac{i}{2\pi}\,\int_{\mathcal{C}_r}\,\rho_2(z)\left
  \{ \mathbf{R}^2(z)[\chi,\mathbf{H}]\mathbf{R}(z)+\mathbf{R}(z)
[\chi,\mathbf{H}]\mathbf{R}^2(z)\right \}\,dz
\end{equation}
we see that a typical contribution to the integrand is an
operator of the form
$\mathbf{R}(z)\chi'\partial_{x_1}\mathbf{R}^2(z)$. We can write ($0<\alpha_1<1$): 
\begin{align}\label{asaptespea}
&\mathbf{R}(z)\chi'\partial_{x_1}\mathbf{R}^2(z)\\
&=\left \{\mathbf{R}(z)e^{-\alpha_1\frac{\delta d}{\lambda}}\right \}
\left \{\chi'
e^{2\alpha_1\frac{\delta d}{\lambda}}\right \}
\left \{e^{-\alpha_1\frac{\delta
      d}{\lambda}}\partial_{x_1}\mathbf{R}(z)e^{\alpha_1\frac{\delta
      d}{\lambda}}\right \}
\left \{e^{-\alpha_1\frac{\delta
    d}{\lambda}}\mathbf{R}(z)\right \} \nonumber
\end{align}
The first and last factors are Hilbert-Schmidt, the second and third
ones are bounded (the support of $\chi'$ is near the origin). Hence
the trace norm is polynomially bounded in $\lambda$, and the
exponential decay of $\rho_2$ will do the rest. 

The proof of \eqref{acinspea} does not contain any new
ingredients. The ideas are the same: introduce \eqref{acincea} in
\eqref{densimatri}, integrate by parts with respect to $z$, propagate
exponential factors over resolvents, and use the
fact that the distance between the support of $\chi'$ and the support
of $\tilde{\Phi}_L$, $1-\Phi_L$ or $\Phi_L'$ is of order $L$. 
\end{proof}

\begin{corollary}\label{limitacurrrent}
The linear response contribution to the current admits the
thermodynamic limit and 
\begin{equation}\label{curent-22}
 \lim_{L\to\infty}\tilde{j}_{L,\eta}=-i\eta\,\int_{-\infty}^0\chi^\prime(\eta
 s)\Tr_{\mathcal{H}}
\left(e^{is\mathbf{H}} V(Q) e^{-is\mathbf{H}}\left[\rho(\mathbf{H}),\Pi_\pm\right]\right)ds.
\end{equation} 
\end{corollary}
\begin{proof}
This is a direct consequence of Corollary \ref{s-conv-H}, 
Proposition \ref{convintrasa}, and the Lebesgue dominated convergence theorem.
\end{proof}

\section{The Adiabatic Limit}

From now on the leads are semiinfinite; the thermodynamic limit has
been taken. Hence in \eqref{adoua2} one has to interpret $L$ as
infinite. 
The adiabatic limit $\eta\rightarrow 0$ in the formula
(\ref{curent-22}) will in fact be an Abel limit once we can show
that the limit $\underset{s\rightarrow\infty}{\lim}e^{is\mathbf{H}}
V(Q) e^{-is\mathbf{H}}$ exists, at least for the strong topology and on a
certain subspace of $\mathbb{B}[\mathcal{H}]$. In view of the definition of
$V(Q)$ (see \eqref{biasul}) we are reduced to studying the limits 
$\underset{s\rightarrow\infty}{\lim}e^{is\mathbf{H}} \Pi_\pm
e^{-is\mathbf{H}}$. 

The idea is to consider the Hamiltonian $\mathbf{H}$ as the
perturbation of a {\it decoupled Hamiltonian}
$\overset{\circ}{\mathbf{H}}_a$ that commutes with the projections
$\Pi_\pm$. We define a {\it decoupled Hilbert space} 
$\overset{\circ}{\mathcal{H}}_a:=(\Pi_-\mathcal{H})\oplus
(\Pi_0\mathcal{H})\oplus (\Pi_+\mathcal{H})$. The decoupled
Hamiltonian is simply $\mathbf{H}$ with two extra Dirichlet boundary
conditions defined by the previous splitting. Let us denote it by 
$\overset{\circ}{\mathbf{H}}_a$. We assume that the internal boundary
defined by the condition $x_1=\pm a$ is smooth enough; it is so in the
cylinder case. It is important to note that:
\begin{equation}\label{prima4}
[\overset{\circ}{\mathbf{H}}_a,\Pi_\pm]=0.
\end{equation}

With these notations and definitions we now have:
\begin{equation}\label{adoua4}
e^{is\mathbf{H}} \Pi_\pm
e^{-is\mathbf{H}}\,=\,e^{is\mathbf{H}}
e^{-is\overset{\circ}{\mathbf{H}}}
\Pi_\pm e^{is\overset{\circ}{\mathbf{H}}} e^{-is\mathbf{H}}.
\end{equation}
Under our assumptions, the Hamiltonian $\mathbf{H}$ in \eqref{atreia4}
is nonnegative and has no singular
continuous spectrum. Let us further assume that $\mathbf{H}$ has no embedded
eigenvalues, and only a finite number of discrete eigenvalues of
finite multiplicity located below the essential spectrum (which can arise from the geometry we chose for
our system $\mathcal{X}_\infty$, \cite{DE}). Then let us
denote by $E_\alpha$ 
the finite dimensional orthogonal projection corresponding to a 
discrete eigenvalue $\lambda_\alpha$ of $\mathbf{H}$, by
$E_\infty$ the projection corresponding to its absolutely continuous
spectrum, and by $E_0:=\underset{\alpha\leq
  N}{\oplus}E_\alpha=1-E_\infty$ the finite dimensional projection on
its discrete spectrum. We can write:
\begin{align}\label{atreia4}
&\Tr_{\mathcal{H}}\left(e^{is\mathbf{H}} V(Q)
  e^{-is\mathbf{H}}\left[\rho(\mathbf{H}),\Pi_\pm\right]\right)=
\Tr_{\mathcal{H}}\left(E_\infty e^{is\mathbf{H}} V(Q) 
e^{-is\mathbf{H}}E_\infty\left[\rho(\mathbf{H}),\Pi_\pm\right]\right)\nonumber \\
&+\sum\limits_{\alpha\leq N}\Tr_{\mathcal{H}}\left(E_\alpha
  e^{is\mathbf{H}} V(Q)
  e^{-is\mathbf{H}}E_\infty\left[\rho(\mathbf{H}),\Pi_\pm\right]\right)\nonumber \\
&+\sum\limits_{\alpha\leq N}\Tr_{\mathcal{H}}\left(E_\infty
  e^{is\mathbf{H}} V(Q) e^{-is\mathbf{H}}E_\alpha\left[\rho(\mathbf{H}),\Pi_\pm\right]\right) \nonumber \\
&+\sum\limits_{\alpha\leq N,\beta\leq N}\Tr_{\mathcal{H}}\left(E_\alpha
  e^{is\mathbf{H}} V(Q) e^{-is\mathbf{H}}E_\beta\left[\rho(\mathbf{H}),\Pi_\pm\right]\right).
\end{align}
Let us show that the last term does not contribute to the current
after the adiabatic limit. We have the identity:
\begin{align}\label{apatra4}
&\Tr_{\mathcal{H}}\left(E_\alpha e^{is\mathbf{H}} V(Q)
  e^{-is\mathbf{H}}E_\beta
  \left[\rho(\mathbf{H}),\Pi_\pm\right]\right)\\
&=
e^{is(\lambda_\alpha-\lambda_\beta)}\Tr_{\mathcal{H}}\left(V(Q)E_\beta[\rho(\mathbf{H}),\Pi_\pm]E_\alpha\right).\nonumber 
\end{align}
Now for $\alpha=\beta$ we have
\begin{equation}\label{acincea4}
E_\alpha[\rho(\mathbf{H}),\Pi_\pm]E_\alpha=\rho(\lambda_\alpha)
(E_\alpha\Pi_\pm E_\alpha-E_\alpha\Pi_\pm E_\alpha)=0.
\end{equation}
For $\alpha\ne\beta$ we use the bound 
$$
\left|\Tr_{\mathcal{H}}\left(V(Q)E_\beta[\rho(\mathbf{H}),\Pi_\pm]
E_\alpha\right)\right|\leq\|V(Q)\|\; \|[\rho(\mathbf{H}),\Pi_\pm]\|_{\mathbb{B}_1}
$$
and the fact that the $s$-integral
$$
\eta\,\int_{-\infty}^0\,\chi^\prime(\eta s)\,e^{is\omega}\,ds=
(\mathcal{F}\chi^\prime)(\omega/\eta)\underset{\eta\rightarrow0}{\longrightarrow}0
$$
for any $\omega\ne 0$, as the Fourier transform of a $L^1$-function on $\mathbb{R}$.

Now let us study the second term in \eqref{adoua4} and prove that it
will also disappear after the adiabatic limit. We have the identity:
\begin{align}\label{asasea4}
&\Tr_{\mathcal{H}}\left(E_\alpha e^{is\mathbf{H}} V(Q)
  e^{-is\mathbf{H}}E_\infty
  \left[\rho(\mathbf{H}),\Pi_\pm\right]\right)\\
&=
\Tr_{\mathcal{H}}\left(V(Q) e^{-is(\mathbf{H}-\lambda_\alpha)}E_\infty
\left[\rho(\mathbf{H}),\Pi_\pm\right]E_\alpha \right)\nonumber 
\end{align}
and observe that (by interchanging the trace with the $s$-integral) we have
\begin{align}\label{asaptea4}
&\Tr_{\mathcal{H}} \left \{V(Q)\,\left(\eta\int_{-\infty}^0\,\chi^\prime(\eta
  s)\, e^{-is(\mathbf{H}-\lambda_\alpha)}\,ds \right )\,
E_\infty\left[\rho(\mathbf{H}),\Pi_\pm\right]E_\alpha \right\}\\
&=\Tr_{\mathcal{H}}\left\{V(Q)\,\left ((\mathcal{F}\chi^\prime)
(\eta^{-1}\frac{}{}(\mathbf{H}-\lambda_\alpha)E_\infty)\right )
\left[\rho(\mathbf{H}),\Pi_\pm\right]E_\alpha
\right\}.\nonumber 
\end{align}
Due to our hypothesis on $\mathbf{H}$ and the definition of 
$E_\infty$, we get:
$$
(\mathbf{H}-\lambda_\alpha)E_\infty\ \geq (\inf(\sigma_{ac}(\mathbf{H}))-\lambda_\alpha)E_\infty,
$$
and since $\mathcal{F}\chi^\prime$ converges to zero at infinity we obtain:
$$
\|(\mathcal{F}\chi^\prime)(\eta^{-1}(\mathbf{H}-\lambda_\alpha)E_\infty)\|_{\mathbb{B}[\mathcal{H}]}\,=\,
\underset{\mu\geq \inf(\sigma_{ac}(\mathbf{H}))-\lambda_\alpha }
{\sup}(\mathcal{F}\chi^\prime)(\mu/\eta)\,\underset{\eta\rightarrow0}
{\longrightarrow}\,0.
$$
The third term in \eqref{atreia4} can be treated in a similar way as
the second one. Therefore only the first term can give a
contribution, and let us identify it. Denote by
$P_{ac}(\overset{\circ}{\mathbf{H}})=\Pi_-\oplus\Pi_+$ the projector on
the absolutely continuous subspace of $\overset{\circ}{\mathbf{H}}$. 
We note that the incoming wave operators
at $-\infty$ associated to the pair of Hamiltonians
$(\overset{\circ}{\mathbf{H}},{\mathbf{H}})$ exist and are
complete. This can be shown in a number of different of ways, but here
we choose to {{} invoke} the invariance principle and the Kato-Rosenblum theorem. Indeed, the function
$-\rho$ is admissible (see Thm. XI.23 [RS III]), and we have: 
\begin{equation}\label{adoua5}
\mathbf{\Omega}_+(\overset{\circ}{\mathbf{H}},\mathbf{H})=
\mathbf{\Omega}_+(-\rho(\overset{\circ}{\mathbf{H}}),-\rho(\mathbf{H}))=
\mathbf{\Omega}_+(-\rho(\overset{\circ}{\mathbf{H}}),-\rho(\overset{\circ}{\mathbf{H}})-\Delta\rho),
\end{equation}
where the operator
$\Delta\rho:=\rho(\mathbf{H})-\rho(\overset{\circ}{\mathbf{H}})$ 
is trace class (more details will be given in the next section). Thus:
\begin{equation}\label{aopta4}
\mathbf{\Omega}_+\,:=\,s-\underset{s\rightarrow-\infty}{\lim}e^{is\mathbf{H}}e^{-is\overset{\circ}{\mathbf{H}}}
P_{ac}(\overset{\circ}{\mathbf{H}}),\quad {\rm Ran}(\Omega_+)=E_\infty.
\end{equation}

Using \eqref{adoua4}, \eqref{aopta4} and the fact that $
V(Q)=V(Q)P_{ac}(\overset{\circ}{\mathbf{H}})$, we obtain:
\begin{align}\label{anoua4}
& \underset{s\rightarrow-\infty}{\lim}\Tr_{\mathcal{H}}\left(E_\infty e^{is\mathbf{H}} 
V(Q)
e^{-is\mathbf{H}}E_\infty\left[\rho(\mathbf{H}),\Pi_\pm\right]\right)\
\nonumber \\
&=\underset{s\rightarrow-\infty}{\lim}\Tr_{\mathcal{H}}\left(E_\infty
  e^{is\mathbf{H}}e^{-is\overset{\circ}{\mathbf{H}}} V(Q)
  e^{is\overset{\circ}{\mathbf{H}}}e^{-is\mathbf{H}}E_\infty\left[\rho(\mathbf{H}),\Pi_\pm\right]\right)\nonumber \\
&=\Tr_{\mathcal{H}}\left(E_\infty
  \mathbf{\Omega}_+V(Q)\mathbf{\Omega}_+^*E_\infty\left[\rho(\mathbf{H}),\Pi_\pm\right]\right).
\end{align}
Using \eqref{anoua4} and \eqref{atreia4} in \eqref{curent-22} we can
write a formula for the adiabatic limit of the linear response current
at $t=0$ as:
\begin{align}\label{azecea4}
  &\tilde{j}_\pm :=\underset{\eta\rightarrow0}{\lim}\,\,\tilde{j}_{\eta}\,=\,i
\Tr_{\mathcal{H}}\left(E_\infty
  \mathbf{\Omega}_+V(Q)\mathbf{\Omega}_+^*E_\infty\left[\rho(\mathbf{H}),\Pi_\pm\right]\right)\\
&=\,v_-i\Tr_{E_\infty\mathcal{H}}\left(\mathbf{\Omega}_+\Pi_-\mathbf{\Omega}_+^*
\left[\rho(\mathbf{H}),\Pi_\pm\right]\right)+v_+i
\Tr_{E_\infty\mathcal{H}}\left(\mathbf{\Omega}_+\Pi_+\mathbf{\Omega}_+^*
\left[\rho(\mathbf{H}),\Pi_\pm\right]\right).\nonumber 
\end{align}
Recalling that the decoupled Hamiltonian commutes with the projections
$\Pi_\pm$ (see \eqref{prima4}), we can write:
$$
\left[\rho(\mathbf{H}),\Pi_\pm\right]\,=\,\left[(\rho(\mathbf{H})-\rho(\overset{\circ}{\mathbf{H}})),\Pi_\pm\right].
$$
Then
we have:
\begin{equation}\label{prima5}
\tilde{j}_\pm\,=\,v_-i\Tr_{E_\infty\mathcal{H}}\left(\mathbf{\Omega}_+
\Pi_-\mathbf{\Omega}_+^*
\left[\Delta\rho,\Pi_\pm\right]\right)+v_+i
\Tr_{E_\infty\mathcal{H}}\left(\mathbf{\Omega}_+\Pi_+\mathbf{\Omega}_+^*
\left[\Delta\rho,\Pi_\pm\right]\right).
\end{equation}
At this point we can show that if $v_+=v_-$, then the current is
zero. Indeed, since 
${\Omega}_+(\Pi_+ +\Pi_-)\mathbf{\Omega}_+^*=1$ on
$E_\infty\mathcal{H}$, and because 
$\Tr_{\mathcal{H}}(\left[\Delta\rho,\Pi_+\right])=0$ we can write: 
\begin{equation}\label{prima2255}
\tilde{j}_+\,=\,(v_--v_+)i\Tr_{E_\infty\mathcal{H}}\left(\mathbf{\Omega}_+
\Pi_-\mathbf{\Omega}_+^*
\left[\Delta\rho,\Pi_+\right]\right)-i\Tr_{E_{\rm pp}\mathcal{H}}(\left[\Delta\rho,\Pi_+\right]).
\end{equation}
But $\Tr_{E_{\rm
    pp}\mathcal{H}}(\left[\Delta\rho,\Pi_+\right])=\Tr(E_{\rm
  pp}(\mathbf{H})\left[\rho(\mathbf{H}),\Pi_+\right])=0$, therefore we
obtain:
\begin{equation}\label{prima55}
\tilde{j}_+\,=\,(v_--v_+)i\Tr_{E_\infty\mathcal{H}}\left(\mathbf{\Omega}_+
\Pi_-\mathbf{\Omega}_+^*
\left[\Delta\rho,\Pi_+\right]\right),
\end{equation}
which is a close analog of the formula given by NESS type approaches (see
\cite{AJPP} and references therein). A remark is in order here. In NESS type
approaches the reservoirs (leads) are suddenly coupled at $t=0$ and the
current is computed in the limit $t \rightarrow \infty$. In general, the
current has an oscillatory component  given by a term similar to the last term
in \eqref{atreia4} \cite{JP1}, so it does not have a definite limit at  $t
\rightarrow \infty$. What it is actually  computed by an averaging procedure, 
is the steady component of the current. As we have seen above, in the
adiabatic switching formalism the current itself has a limit as $\eta
\rightarrow 0$. This fact confirms (in our particular setting) the 
basic feature of the adiabatic switching, namely that it ``kills'' oscillatory
terms of various physical quantities.

\section{The Landauer-B{\"u}ttiker formula}

We will now compute the trace appearing in equation \eqref{prima55}
using a spectral representation of $\overset{\circ}{\mathbf{H}}$. By
doing this we will obtain the L-B formula. 
 
Let us enumerate a few technical results needed in order to justify
the computations which will follow next. In the next subsection we
will only give the proof of \eqref{adoua25}, since the
other points are just well known facts which can be found in any
standard book treating the stationary scattering theory.

\begin{proposition}\label{propteh} Consider some $n>1/2$. 
\begin{enumerate}
\item The operator
  $\Delta\rho=\rho(\mathbf{H})-\rho(\overset{\circ}{\mathbf{H}})$ is
  trace class, and if $\langle \x\rangle :=\sqrt{1+x_1^2}$ we have:
\begin{equation}\label{adoua25}
\langle\cdot\rangle^n (\Delta\rho) \langle\cdot\rangle^n \in \mathbb{B}(\mathcal{H}).
\end{equation}
\item Denote by $E_0\geq 0$ the infimum of the essential 
spectrum of $\overset{\circ}{\mathbf{H}}$. Then 
$\sigma_{\rm ac}(\overset{\circ}{\mathbf{H}})=[E_0,\infty)$. Moreover,
the spectrum of $\mathbf{H}$ in $[E_0,\infty)$ is absolutely
continuous, with finitely many (possibly none) eigenvalues. The set of
thresholds $\mathcal{T}$ (points where we do not have a Mourre estimate, see
\cite{mourre}) is discrete and contains infinitely many points. If $n$ is large enough then the mapping 
\begin{align}\label{atreia5}
& A:[E_0,\infty)\setminus \mathcal{T}\mapsto \mathbb{B}(\mathcal{H}), \\
&A(\lambda):=
\lim_{\epsilon\searrow
  0}\: \langle\cdot\rangle^{-n}(\mathbf{H}-\lambda-i\epsilon)^{-1}\langle\cdot\rangle^{-n}\nonumber 
\end{align}
is continuously differentiable. 
\item For $n$
  large enough, the
  mapping 
\begin{align}\label{apatra5}
& A_\rho:[0,\rho(E_0)]\setminus \rho(\mathcal{T})\mapsto \mathbb{B}(\mathcal{H}), \\
&A_\rho(t):=
\lim_{\epsilon\searrow
  0}\: \langle\cdot\rangle^{-n}(\rho(\mathbf{H})-t-i\epsilon)^{-1}\langle\cdot\rangle^{-n}\nonumber 
\end{align}
is continuously differentiable. 
\item We can choose a family of generalized eigenfunctions for the
  absolutely continuous part of the operator $\overset{\circ}{\mathbf{H}}$, which in the case of
 our cylinder-like semiinfinite leads can be {{} explicitly} written down.
 The multiplicity of the absolutely continuous spectrum is finite when
 the energy is restricted to compacts, but it is not constant and
 increases with the energy. A given channel will be indexed by
 $(\alpha,\sigma)$, where $\sigma\in\{\pm 1\}$ shows the left/right
 lead, and $\alpha$ indexes all other quantum numbers at a given
 energy.  Hence
 for $n>1/2$
 we have :
\begin{equation}\label{acincea5}
\langle \cdot\rangle ^{-n}\overset{\circ}{\varphi}_E^{(\alpha,\pm)}\in
\mathcal{H},\quad\overset{\circ}{\mathbf{H}}\overset{\circ}{\varphi}_E^{(\alpha,\pm)}
"="E\overset{\circ}{\varphi}_E^{(\alpha,\pm)}, \quad E \in[E_0,\infty)\setminus
  \mathcal{T}.
\end{equation}
\item With this choice, we can define the corresponding generalized
  eigenfunctions of $\mathbf{H}$ either with the help of the limiting
  absorption principle:
\begin{align}\label{asasea5}
&\langle \cdot\rangle ^{-n}{\varphi}_E^{(\alpha,\pm)}= 
\langle \cdot\rangle ^{-n}\overset{\circ}{\varphi}_E^{(\alpha,\pm)}
-A_\rho(\rho(E))\left \{\langle \cdot\rangle ^{n}\Delta\rho \langle
\cdot\rangle ^{n}\right \}
\left \{\langle \cdot\rangle ^{-n}\overset{\circ}{\varphi}_E^{(\alpha,\pm)}\right
\},\\
& E \in[E_0,\infty)\setminus
  \mathcal{T},\nonumber 
\end{align}
or as solutions of the Lippmann-Schwinger equation (with the usual
abuse of notation):
\begin{align}\label{asaptea5}
&{\varphi}_E^{(\alpha,\pm)}= \overset{\circ}{\varphi}_E^{(\alpha,\pm)}
-(\rho(\overset{\circ}{\mathbf{H}})-\rho(E)-i0_+)^{-1}(\Delta\rho )
{\varphi}_E^{(\alpha,\pm)},\\
& E \in[E_0,\infty)\setminus
  \mathcal{T}.\nonumber 
\end{align}
If $n$ is large enough, then the map 
\begin{equation}\label{aopta5}
[E_0,\infty)\setminus
  \mathcal{T}\ni E\mapsto \langle \cdot\rangle ^{-n}{\varphi}_E^{(\alpha,\pm)}
  \in \mathcal{H}
\end{equation}
is continuously differentiable. 
\end{enumerate}
\end{proposition}  
Let us now use this proposition in order to compute the current. The
main idea is to compute the integral kernels of the two trace class 
operators in \eqref{prima5}, and then to compute the trace as the
integral of the diagonal values of their kernels. We use for this the
generalized eigenvalues of $\mathbf{H}$. Let us choose two energies
$E,E'\in \sigma_{\rm ac}(\mathbf{H})\setminus \mathcal{T}$. Denote by
$(\alpha,\pm)$ and $(\alpha',\pm)$ the quantum numbers describing the spectral 
multiplicity of $\mathbf{H}$ at $E$ and $E'$. 
Then in the spectral representation
of $\mathbf{H}$, an operator like $E_\infty\mathbf{\Omega}_+
\Pi_-\mathbf{\Omega}_+^* (\Delta\rho )\Pi_+$ will have the integral
kernel (here $\sigma,\sigma'\in \{+,-\}$ indicate the leads): 
\begin{equation}\label{anoua5}
\mathcal{I}(E,\alpha,\sigma; E',\alpha',\sigma')=
\delta_{\sigma,+}
\langle (\Delta\rho)
\Pi_+{\varphi}_{E'}^{(\alpha',\sigma')},{\varphi}_E^{(\alpha,+)}\rangle .
\end{equation}
In deriving this equation we formally used the "intertwining" property 
\begin{equation}\label{azecea5}
\mathbf{\Omega}_+^*\varphi_E^{(\alpha,\sigma)}\;"="\;
\overset{\circ}{\varphi}_E^{(\alpha,\sigma)}
\end{equation}
since the wave operators are unitary between the absolutely continuous
subspaces of $\overset{\circ}{\mathbf{H}}$ and $\mathbf{H}$. The
scalar product is in fact a duality bracket between weighted
spaces, and we used \eqref{atreia5}. This integral kernel is jointly
continuous in its energy variables outside the set of
{{} thresholds}, due to \eqref{aopta5}. Thus when we compute the trace of  
$E_\infty\mathbf{\Omega}_+
\Pi_-\mathbf{\Omega}_+^* (\Delta\rho )\Pi_+$, we may write:
\begin{equation}\label{aunspea5}
{\rm Tr}E_\infty\mathbf{\Omega}_+
\Pi_-\mathbf{\Omega}_+^* (\Delta\rho )\Pi_+=\lim_
{S\to \sigma_{\rm ac}(\mathbf{H})\setminus \mathcal{T}}
\int_S\sum_{\alpha,\sigma}\mathcal{I}(E,\alpha,\sigma; E,\alpha,\sigma)dE,
\end{equation}
where $S$ denotes compact sets included in 
$\sigma_{\rm ac}(\mathbf{H})\setminus \mathcal{T}$, and the limit
means that the Lebesgue measure of $\sigma_{\rm
  ac}(\mathbf{H})\setminus 
(\mathcal{T}\cup S)$ goes to zero. 

Now after undoing the commutator in \eqref{prima55}, we compute the
traces in the way we described above, keeping in mind that the energy integral has
to be understood as a limit avoiding the thresholds. Thus we obtain:

\begin{align}\label{adoispea5}
&\tilde{j}_+ =-\,i\int_{E_0}^{\infty}\sum_{\alpha} 
\left\{(v_--v_+)\langle
  [\Delta\rho,\Pi_+]\varphi_E^{(\alpha,-)},\varphi_E^{(\alpha,-)}\rangle 
\right\}\,dE\nonumber \\
&=2(v_--v_+)\int_{E_0}^{\infty}\sum_{\alpha} 
\left\{\Im \langle
  \Delta\rho \Pi_+\varphi_E^{(\alpha,-)},\varphi_E^{(\alpha,-)}\rangle
\right\}\,dE.
\end{align}
Using the Lippmann-Schwinger equation \eqref{asaptea5} we obtain 
$$\Delta\rho \Pi_+\varphi_E^{(\alpha,-)}=-\Delta\rho \Pi_+(\rho(\overset{\circ}{\mathbf{H}})-\rho(E)-i0_+)^{-1}(\Delta\rho )
{\varphi}_E^{(\alpha,-)},$$ 
and we insert this in \eqref{adoispea5}. Using the generalized
eigenfunctions of $\overset{\circ}{\mathbf{H}}$, the limiting
absorption principle, the various regularity properties listed in
Proposition \ref{propteh}, and the Sokhotskii-Plemelj formula, we obtain:
\begin{align}\label{atreispea5}
& \Im \left \langle \Pi_+(\rho(\overset{\circ}{\mathbf{H}})-\rho(E)-i0_+)^{-1}(\Delta\rho )\varphi_E^{(\alpha,-)}
  ,(\Delta\rho )\varphi_E^{(\alpha,-)}\right \rangle \\
&=\frac{\pi}{\rho'(E)}\sum_{\alpha'}\left \vert \left \langle (\Delta\rho
)\varphi_E^{(\alpha,-)},\overset{\circ}{\varphi}_E^{(\alpha',+)}\right
\rangle \right
\vert^2.\nonumber 
\end{align}
Thus the current reads as: 
\begin{align}\label{apaispea5}
&\tilde{j}_+ =2\pi (v_+-v_-)\int_{E_0}^{\infty}\sum_{\alpha,\alpha'} \frac{1}{\rho'(E)}\left \vert \left \langle (\Delta\rho
)\varphi_E^{(\alpha,-)},\overset{\circ}{\varphi}_E^{(\alpha',+)}\right
\rangle \right
\vert^2 dE.
\end{align}
We now have to relate the above integrand with the
scattering matrix. We know that the $S$ matrix commutes with
$\overset{\circ}{\mathbf{H}}$, and so does the $T$ matrix defined as
$\frac{1}{2\pi i}(1-S)$. In the spectral representation of
$\overset{\circ}{\mathbf{H}}$ induced by $
  \{\overset{\circ}{\varphi}_E^{(\alpha,\sigma)}\}_{\alpha,\sigma}$, the $T$ operator is a direct integral
with a fiber which is a finite dimensional matrix. Now using the
correspondence principle (see \eqref{adoua5}), and formula ${\rm
  (4)}_\pm$, page 233 in \cite{Yafaev} we can write (see also Thm. XI.42 in
\cite{RS3}):
\begin{equation}\label{prima6}
T_{\alpha,\sigma;\alpha',\sigma'}(E)=\frac{1}{\rho'(E)} \left \langle (\Delta\rho
)\varphi_E^{(\alpha',\sigma')},\overset{\circ}{\varphi}_E^{(\alpha,\sigma)}\right
\rangle.
\end{equation}
The Landauer-B\"uttiker formula becomes:
\begin{align}\label{apaispea77}
&\tilde{j}_+ =2\pi (v_+-v_-)\int_{E_0}^{\infty}\sum_{\alpha,\alpha'} \rho'(E)\left \vert T_{\alpha',+;\alpha,-}(E)\right
\vert^2 dE.
\end{align}
We can now state the main result of our paper:
\begin{theorem}\label{teoremaprincipala}
Assume that the system has $N$ leads. Then the charge variation of
lead $f=1$ defined after the thermodynamic and adiabatic limits, in
the linear response regime, reads as:
\begin{equation}\label{kkk1}
j_1=\sum_{f=2}^N(v_f-v_1)i{\rm Tr}_{E_\infty\mathcal{H}}\left
  (\Omega_+\Pi_f\Omega_+^*[\Delta\rho,\Pi_1]\right ).
\end{equation}
In terms of transmittance between leads we have:
\begin{align}\label{adouazecea5}
j_1 =2\pi \sum_{f=2}^N(v_1-v_f)\int_{E_0}^{\infty}\sum_{\alpha,\alpha'} \rho'(E)\left \vert T_{\alpha',1;\alpha,f}(E)\right
\vert^2 dE.
\end{align}

\end{theorem}

\subsection{A continuity equation for the current}

If we go back to \eqref{aopta3} and instead of $\Pi_\pm$ we use some
smaller projections $\hat{\Pi}_\pm$ which have the property that
$\hat{\Pi}_\pm  \Pi_\pm=\hat{\Pi}_\pm$ (but still corresponding to an
infinitely long portion of a lead), then we can again define a
current and obtain an expression as in \eqref{curent-2}, where
$\Pi_\pm$ are replaced by $\hat{\Pi}_\pm$. The thermodynamic and
adiabatic limits can be performed in the same manner, and we would
obtain a formula very similar to \eqref{kkk1}:

\begin{equation}\label{kkk2}
\hat{j}_1=\sum_{f=2}^N(v_f-v_1)i{\rm Tr}_{E_\infty\mathcal{H}}\left
  (\Omega_+\Pi_f\Omega_+^*[\rho(\mathbf{H}),\hat{\Pi}_1]\right ).
\end{equation}
Note that $\overset{\circ}{\mathbf{H}}$ is still defined with respect
to the old decomposition imposed by $\Pi$'s, and NOT by
$\hat{\Pi}_1$. That is $\overset{\circ}{\mathbf{H}}$ does NOT commute
with $\hat{\Pi}_1$. Another important observation is that
$[\rho(\mathbf{H}),\hat{\Pi}_1]$ is already trace class even before
inserting something commuting with $\hat{\Pi}_1$, as we can see from the
proof of Proposition \ref{convintrasa}. If we {{} subtract} $\hat{j}_1$
from $j_1$ we obtain (denote by $\pi_1:=\Pi_1-\hat{\Pi}_1$):
\begin{equation}\label{kkk3}
j_1-\hat{j}_1=\sum_{f=2}^N(v_f-v_1)i{\rm Tr}_{E_\infty\mathcal{H}}\left
  (\Omega_+\Pi_f\Omega_+^*[\rho(\mathbf{H}),\pi_1]\right ).
\end{equation}
Now $\pi_1$ corresponds to the multiplication with a function with
compact support. Using again the methods of Proposition
\ref{convintrasa}, one can show that $\rho(\mathbf{H})\pi_1$ is
trace-class (one can write it as a product of two Hilbert-Schmidt
operators). Then 
\begin{align}\label{prima77}{\rm Tr}_{E_\infty\mathcal{H}}\left
  (\Omega_+\Pi_f\Omega_+^*[\rho(\mathbf{H}),\pi_1]\right )={\rm Tr}\left
  (\Omega_+\Pi_f\Omega_+^*\rho(\mathbf{H})\pi_1 \right )-{\rm Tr}\left
  (\Omega_+\Pi_f\Omega_+^*\pi_1\rho(\mathbf{H}) \right ).
\end{align}
Let us show the operator $\rho(\mathbf{H})\Omega_+\Pi_f\Omega_+^*\pi_1$ is
also trace class: note first that
$[\Omega_+\Pi_f\Omega_+^*,\rho(\mathbf{H})]=0$ (this is due to the
intertwining property of wave operators and due to
$[\overset{\circ}{\mathbf{H}},\Pi_f]=0$); second, we can put near $\pi_1$
enough resolvents in order to obtain the product of two
Hilbert-Schmidt operators. Then let us prove the trace {{} cyclicity}:
\begin{align}\label{atreia77}{\rm Tr}\left
  (\Omega_+\Pi_f\Omega_+^*\pi_1\rho(\mathbf{H}) \right )={\rm Tr}\left (
  \rho(\mathbf{H})\Omega_+\Pi_f\Omega_+^*\pi_1 \right ).
\end{align}
Indeed, take $\pi_2$ with compact support such that $\pi_2\pi_1=\pi_1$.
The usual trace {{} cyclicity} gives:
\begin{align}\label{apatra77}{\rm Tr}\left
  (\Omega_+\Pi_f\Omega_+^*\pi_1\rho(\mathbf{H}) \right )={\rm Tr}\left
  (\Omega_+\Pi_f\Omega_+^*\pi_1\pi_2 \rho(\mathbf{H}) \right )={\rm Tr}\left (
  \pi_2\rho(\mathbf{H})\Omega_+\Pi_f\Omega_+^*\pi_1 \right ).
\end{align}
Now take back $\pi_2$ to the right and \eqref{atreia77} is proved. Thus we can write:
\begin{align}\label{adoua77}{\rm Tr}\left
  (\Omega_+\Pi_f\Omega_+^*\pi_1\rho(\mathbf{H}) \right )={\rm Tr}\left (
  \rho(\mathbf{H})\Omega_+\Pi_f\Omega_+^*\pi_1 \right )={\rm Tr}\left
  (\Omega_+\Pi_f\Omega_+^*\rho(\mathbf{H}) \pi_1\right )
\end{align}
where in the second equality we used that
$[\Omega_+\Pi_f\Omega_+^*,\rho(\mathbf{H})]=0$. Use this in
\eqref{prima77} and obtain
$$j_1=\hat{j}_1,$$
which shows that the current is the same no matter where we measure it
on the lead.

\subsection{Sketching the proof of \eqref{adoua25}}
\begin{proof}  

The main idea is to use the geometric perturbation theory, in the same
spirit as in Proposition \ref{convintrasa}. With the functions
introduced in \eqref{adoua} and \eqref{atreia}, choose $L$ large
enough such that
$\mathbf{H}\tilde{\Phi}_L=\overset{\circ}{\mathbf{H}}\tilde{\Phi}_L$.
This means that the "sample" is contained in the domain where
$\Phi_L=1$. Then keep $L$ fixed. Now for $z\in \mathbb{C}\setminus\R$
we define the parametrix:
\begin{equation}\label{adoua6}
S(z):=\Phi_L({\mathbf{H}}-z)^{-1}\Phi_L+\tilde{\Phi}_L(\overset{\circ}{\mathbf{H}}-z)^{-1}\tilde{\Phi}_L.
\end{equation}
We have that the range of $S(z)$ is mapped into the domain of
$\mathbf{H}$ and:
\begin{align}\label{atreia6}
({\mathbf{H}}-z)S(z)&=:1+T(z),\\ 
 T(z)&= [-2(\nabla\Phi_L)\nabla
-(\Delta\Phi_L)]({\mathbf{H}}-z)^{-1}\Phi_L\nonumber \\
&+[-2(\nabla\tilde{\Phi}_L)\nabla
-(\Delta\tilde{\Phi}_L)](\overset{\circ}{\mathbf{H}}-z)^{-1}\tilde{\Phi}_L.\nonumber
\end{align}
Similarly as in \eqref{acincea} we can write: 
\begin{equation}\label{apatra6}
({\mathbf{H}}-z)^{-1}=S(z)-({\mathbf{H}}-z)^{-1}T(z).
\end{equation}
Thus:
\begin{align}\label{acincea6}
({\mathbf{H}}-z)^{-1}-(\overset{\circ}{\mathbf{H}}-z)^{-1}&=\Phi_L({\mathbf{H}}-z)^{-1}\Phi_L+
(\tilde{\Phi}_L-1)(\overset{\circ}{\mathbf{H}}-z)^{-1}\tilde{\Phi}_L\nonumber
\\
&+(\overset{\circ}{\mathbf{H}}-z)^{-1}(\tilde{\Phi}_L-1)-({\mathbf{H}}-z)^{-1}T(z).
\end{align}
Then using the Cauchy integral representation as in \eqref{aopta}, we
see that the operator $\Delta\rho$ can be expressed as an integral
where the integrand consists of several terms, and we can write each
of them as a product of two Hilbert-Schmidt operators, due to the
support properties of our cut-off functions and the exponential decay
of various integral kernels. In order to prove \eqref{adoua25} we need
to propagate some exponential decay near both polynomials
$\langle\cdot\rangle^n$, and this can be done as in Proposition
\ref{convintrasa}. We do not give further details. \end{proof}

\vspace{0.5cm}

\noindent{\bf Acknowledgements.} Part of this work was done during a visit of G. Nenciu
at the Department of Mathematical Sciences, Aalborg University;
 both hospitality and financial support are gratefully acknowledged.
H. Cornean acknowledges support from the Danish F.N.U. grant
{\it Mathematical Physics and Partial Differential Equations}. G. Nenciu was partially supported by 
CEEX Grant D11-45/2005. R. Purice was partially supported by CEEX Grant  2-CEx06-11-97/2006.

\end{document}